\documentclass[runningheads]{llncs}

\usepackage{enumitem}
\usepackage{amssymb}
\usepackage{amsmath}
\usepackage{bm}

\usepackage{tikz}
\usetikzlibrary{arrows.meta,arrows}
\usetikzlibrary{shapes.geometric}
\usetikzlibrary{positioning}
\usetikzlibrary{backgrounds}

\newcommand{\frog}{\textsc{FRoG}}
\newcommand{\qedef}{\hfill$\diamond$}

\begin{document}

\title{The Complexity of Sequential Routing Games}

\author{Anisse Ismaili \inst{1,2}}

\authorrunning{A. Ismaili}

\institute{Multi-agent Optimization Team, RIKEN AIP Center, Tokyo, Japan\and
Yokoo Laboratory, Kyushu University, Fukuoka, Japan\\
anisse.ismaili@gmail.com}

\maketitle            

\begin{abstract}
We study routing games where every agent sequentially decides her next edge when she obtains the green light at each vertex. Because every edge only has capacity to let out one agent per round, an edge acts as a FIFO waiting queue that causes congestion on agents who enter. Given $n$ agents over $|V|$ vertices, we show that for one agent, approximating a winning strategy within $n^{1-\varepsilon}$ of the optimum for any $\varepsilon>0$, or within any polynomial of $|V|$, are PSPACE-hard. Under perfect information, computing a subgame perfect equilibrium (SPE) is PSPACE-hard and in FPSPACE. Under imperfect information, deciding SPE existence is PSPACE-complete.
\end{abstract}

\section{Introduction}



%
%
%
%

Models for routing games stem from applications in road traffic \cite{wardrop1952road} and packet routing via Internet Protocol \cite{koutsoupias1999worst}. Many agents share a same routing network (a graph of vertices and edges) without any central authority, and each agent decides a path, with the individualistic goal to minimize her delay. However, by using the shared resources of the network, agents incur congestion to one another: they play a game.
\emph{Static} routing games, where every individual's decided path congests its edges all at once, made history as the first application for the Price of Anarchy (PoA), a measure for the loss of efficiency due to selfish behaviors \cite{roughgarden2002bad,roughgarden2005selfish,christodoulou2005price,awerbuch2005price,nisan2007algorithmic,roughgarden2009intrinsic}.
Static routing games, being potential games, also have the very satisfying property that a best-response dynamics always converges to a pure Nash equilibrium (PNE).

However, static models do not reflect the temporality of travels. Indeed, an agent, when on one edge of her path, should not be congesting the other edges. In routing games \emph{over time}, every agent travels through time as well as along the edges of her path. This more recent model was also deeply studied \cite{koch2009nash,anshelevich2009equilibria,hoefer2009competitive,werth2014atomic,harks2016competitive,ismaili2017,Harks:2018:CPR:3182630.3184137}. 
Although natural, adding such temporality in delays calculations increases the complexity of games: pure-strategy Nash equilibria may not exist; computing a best-response or an equilibrium are computationally intractable; the PoA can be large. 

Still, a last temporal aspect is missing from routing games over time. Although the locality of congestion was related with the time at which it occurs, agents might also make decisions along time. \emph{Sequential} routing games were introduced even more recently \cite{cao2017ec,cao2017arxiv}. Every agent, when she reaches any node, is allowed to observe previous actions, in order to decide the next course of actions, like the next edge on her path. 
In this sequential setting, winning strategies and subgame perfect equilibrium (SPE) are standard concepts for individual rationality and a stable game outcome.

\textbf{Related Work.}
There is a rich literature on \emph{routing games over time}.
%
In \cite{harks2006competitive,harks2009competitive}, 
there are multiple asymmetric commodities that are routed in a sequence.
The cost of edges is an affine function of the load.
For splittable commodities, the PoA is almost four, and for unsplittable ones it is $3+2\sqrt{2}$.
The best-response problem is NP-hard.
%
\cite{koch2009nash,koch2011nash}
generalizes deterministic queuing into competitive flows over time.
In this non-atomic symmetric model, 
an edge's outflow is determined by its capacity, 
and also has a fixed transit delay.
An iterative algorithm is proposed, where $\varepsilon$-Nash flows converge to a Nash flow.
%
\cite{anshelevich2009equilibria} studies routing games with non-atomic asymmetric agents. 
While for symmetric agents, a PNE is guaranteed to exist and efficiently computable, 
in the asymmetric case and under FIFO, an equilibrium may not exist. 
The PoA is bounded from below by the number of vertices. 
%
\cite{hoefer2009competitive,hoefer2011competitive}
proposes temporal network congestion games. 
Every edge has a processing speed $a_e\in\mathbb{R}_{>0}$, and various local policies are studied.
Under FIFO, edge $e$ processes one agent in time $a_e$, while the other agents wait. 
In the unweighted symmetric case, despite the NP-hardness of computing a best-response, a PNE can always be efficiently computed. Otherwise, an equilibrium may not exist.
%
While costs in \cite{koch2009nash} where rationals, \cite{werth2014atomic} focuses on atomic agents and integer time-steps. Every edge is associated with a capacity (maximal flow-per-time-step) and a fixed transit delay. 
A PNE may not exist; 
the best-response problem is NP-complete; 
PNE verification is coNP-complete;
PNE existence is at least NP-hard;
a bound is provided on the PoA.
%
%
\cite{harks2016competitive,Harks:2018:CPR:3182630.3184137} studies a model related to \cite{werth2014atomic}.
The focus is on local and overall priorities on the edges. (Crossovers may occur.)
Several bounds on the price of stability and PoA are shown.
It is APX-hard to compute optimal priority lists.
Under local priorities, the best-response and PNE problem are NP-hard.
\cite{ismaili2017} studies a model related to \cite{harks2016competitive} and similar to \cite{werth2014atomic} under FIFO.
Best-responses are inapproximable, and deciding whether a PNE exists is complete for class $\Sigma_2^P$.

%
Concerning \emph{sequential routing games},
\cite{cao2017arxiv,cao2017ec} are the first to propose a model where agents take new decisions on each vertex. In this very interesting setting, when several agents arrive on the same edge during the same round, ties are broken by an order on incoming edges. It allows for plays that do not depend on agents' IDs.
It also guarantees the existence of a subgame perfect equilibrium (when all agents share the same sink), in a constructive way: the iterative dominating path profile algorithm \cite[Alg. 1]{cao2017arxiv} outputs a path profile induced by an SPE. It is also reminded that an SPE might not exist when one uses an overall order on agents for tiebreaking \cite[Ex. 2]{cao2017arxiv}. 

Besides, \cite{papadimitriou2010new} studies a new model where decisions are made by nodes instead of flows.
There is also a diverse literature on repeated routing games,
where sequentiality comes from a repetition of static routing games (see e.g. \cite{blum2006routing,roughgarden2009intrinsic,chien2011convergence} or \cite{monnot2017routing} in real life).
While it is well known that problem QSAT (also known as TQBF) is PSPACE-complete \cite{garey1979computers}, the Canadian Traveler Problem \cite{PAPADIMITRIOU1991} was an inspiration for settling our complexities.

\textbf{Our contributions.}
Here, we study sequential routing games that are played on a digraph, round after round.
W.l.o.g. on negative results, we assume that every edge has capacity and length one.
Every agent has to travel from her source-vertex to her sink-vertex in as few rounds as she can.
Congestion is modeled by a very natural FIFO queuing policy on every edge: 
agents who enter the edge are queued, and an edge lets out one agent per round.
At every round, the top agents of all waiting lists pop out from their current edges, decide their next edges and queue there.
More precisely, we study an imperfect information setting where the top agents decide their next edges simultaneously,
and a perfect information setting where they decide in an ordered sequence.
We study the computational complexity of standard rationality and stability concepts, 
and show the following results in both sequential and simultaneous action settings.
\begin{enumerate}[leftmargin=16mm]
\item[\bf Th. 1] Deciding whether an agent has a strategy that will guarantee her total delay below a given threshold is PSPACE-complete.
\item[\bf Th. 2] Let $n$ be the number of agents and $|V|$ the number of vertices. Approximating an agent's optimal strategy within $n^{1-\varepsilon}$ for any $\varepsilon>0$, and within any polynomial of $|V|$ are both PSPACE-hard. 
\item[\bf Th. 3] Under perfect information, computing a subgame perfect equilibrium is PSPACE-hard and in FPSPACE. Under imperfect information, deciding SPE existence is PSPACE-complete.
\end{enumerate}

\vspace*{-6mm}

\section{Preliminaries}
\vspace*{-2mm}

We introduce relevant notation on routing games and sequential games.
Without loss of generality on negative results, edges have unitary length and capacity.
\footnote{A similar assumption (unitary length and capacity) is made in \cite{cao2017arxiv}. Since the present hardness results still hold in this particular case, it means that computational complexity does not come from numbers defining length and capacity.}

\begin{definition}[Sequential First-in-first-out Routing Game (FRoG)]~\\
A \frog\ is characterized by a  tuple $\left(G=(V,E),N,(s_i,s_i^\ast)_{i\in N}\right)$ where:
\begin{itemize}
\item $G=(V,E)$ is a finite digraph with vertex set $V$ and edges\footnote{An edge $e=(u_e,v_e)$ is a couple of vertices. Its tail (resp. head) is $u_e$ (resp. $v_e$).} $E\subseteq V\times V$,
\item finite set $N=\{1,\ldots,n\}$ is the set of agents, and
\item for every agent $i$, vertex $s_i\in V$ is her source and $s_i^\ast\in V\setminus\{s_i\}$ her sink. \qedef
\end{itemize}
\end{definition}
\noindent 
For every edge $e$, set $F(e)\subseteq E$ denotes the \emph{successors} of $e$: all the edges whose tail is the head of $e$.\footnote{Vertices are mere connections. Edges are the only decided and time-costly resources.}
For every agent $i$, starting from her source $s_i$, the goal is to travel a path\footnote{A path is a finite list of edges such that for two consecutive edges the head of the former is the tail of the later. We allow edge repetitions and directed cycles. We assume that there always exists a path from $s_i$ to $s_i^\ast$.} $\pi_i$ that reaches her sink $s_i^\ast$, by deciding a successor edge each time she reaches a head.
Every edge $e$ acts as a FIFO \emph{queuing list}\footnote{A queuing list admits two operations: queuing and popping. Under FIFO ordering, queuing adds an agent at the end of the list, and popping takes the first agent.} $Q_e$ on agents:
every agent who enters $e$ is queued in $Q_e$.

A game is played sequentially over \emph{rounds} $r\in\mathbb{N}$.
Here is the main loop: once per round, every non-empty list $Q_e$ simultaneously pops its first agent $i$, who simultaneously decides her next edge $e'$ in successors $F(e)$, and enters waiting queue $Q_{e'}$. If the head of edge $e$ is already $i$'s sink $s_i^\ast$, she simply exits the system.
(If two agents arrive on the same edge at the same time, rules are precised later.)
Congestion comes from the fact that edges let out at most one agent per round.

If the game ends in a finite number of rounds\footnote{An agent can cycle or decide to get stuck somewhere with no way to reach her sink.},
every agent reached her sink,
and the outcome is
a path-profile $(\pi_1,\ldots,\pi_n)\in\mathcal{P}_1\times\ldots\times\mathcal{P}_n$, which we denote in bold by $\bm{\pi}\in\bm{\mathcal{P}}$, and where $\mathcal{P}_i$ denotes the set of finite paths from agent $i$'s source $s_i$ to her sink $s_i^\ast$.
Informally, the \emph{total-delay} (defined later) that an agent $i$ seeks to minimize is $C_i(\bm{\pi})=\sum_{\ell=1}^{|\pi_i|}w_i(\bm{\pi},\ell)$ 
where $w_{i}(\bm{\pi},\ell)\in\mathbb{N}_{\geq 1}$ is agent $i$'s position in waiting list $Q_{\pi_i(\ell)}$ when she enters edge $\pi_i(\ell)$.

\begin{example}
In this \frog, vertices are circles, edges (with unitary capacities and delays) are arrows. 
Agent $1$ (resp. $2$) start from $a$ (resp. $b$). Both go to $k$, under simultaneous actions and agent priority $2\succ 1$ (see tiebreaking rule (RO)).
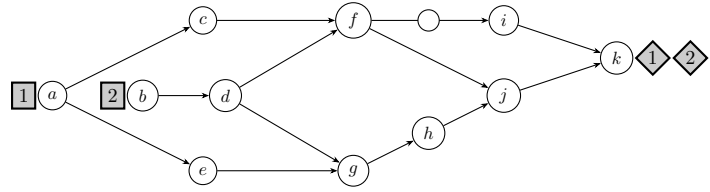
\begin{center}
\begin{tikzpicture}
\draw[]
	node[circle,draw=black,scale=0.67] (a) at (0,0) {$a$}
	node[circle,draw=black,scale=0.67] (b) at (1.2,0) {$b$}
	node[circle,draw=black,scale=0.67] (c) at (2,+1) {$c$}
	node[circle,draw=black,scale=0.67] (d) at (2.3,0) {$d$}
	node[circle,draw=black,scale=0.67] (e) at (2,-1) {$e$}
	node[circle,draw=black,scale=0.67] (f) at (4,+1) {$f$}
	node[circle,draw=black,scale=0.67] (g) at (4,-1) {$g$}
	node[circle,draw=black,scale=0.67] (h) at (5,-0.5) {$h$}
	node[circle,draw=black,scale=0.67] (i) at (6,1) {$i$}
	node[circle,draw=black,scale=0.67] (j) at (6,0) {$j$}
	node[circle,draw=black,scale=0.67] (k) at (7.5,0.5) {$k$};
\draw[]
	(a) edge[-{Stealth[scale=0.6]}] (c)
	(a) edge[-{Stealth[scale=0.6]}] (e)
	(b) edge[-{Stealth[scale=0.6]}] (d)
	(c) edge[-{Stealth[scale=0.6]}] (f)
	(d) edge[-{Stealth[scale=0.6]}] (f)
	(d) edge[-{Stealth[scale=0.6]}] (g)
	(e) edge[-{Stealth[scale=0.6]}] (g)
	(f) edge[-{Stealth[scale=0.6]}] (i)
	(f) edge[-{Stealth[scale=0.6]}] (j)
	(g) edge[-{Stealth[scale=0.6]}] (h)
	(h) edge[-{Stealth[scale=0.6]}] (j)
	(i) edge[-{Stealth[scale=0.6]}] (k)
	(j) edge[-{Stealth[scale=0.6]}] (k);
\node[rectangle,draw=black,fill=black!20, thick,left = 0.25mm of a,scale=0.75]{$1$};
\node[rectangle,draw=black,fill=black!20, thick,left = 0.25mm of b,scale=0.75]{$2$};
\node[diamond,draw=black,fill=black!20, thick,right = 0.25mm of k,scale=0.75, inner sep=0.5mm]{$1$};
\node[diamond,draw=black,fill=black!20, thick,right = 5.25mm of k,scale=0.75, inner sep=0.5mm]{$2$};
\node[circle,draw=black,fill=white,inner sep=1mm] at (5,1) {~};
\end{tikzpicture}
\end{center}
\end{example}

\begin{definition}[Consecutive Configurations and History]
For every round $r\in\mathbb{N}$,
{configuration} $Q(r)=(Q_e(r))_{e\in E}$ is defined as the ordered contents of every waiting list, before tops are popped  to end round $r$ and agents decide their next edges.\footnote{While $Q(0)$ is empty, for $r\geq 1$, configuration $(Q_e(r))_{e\in E}$ is a partition of the agents still in the system. Agents on top of queues in $Q(r)$ decide a new edge at round $r$.} 
A history  $H(r)=(Q(0),\ldots,Q(r))$ is a sequence of consecutive configurations:
from history $H(r)$, if configuration $Q(r)$ is non-empty, we obtain a {consecutive configuration} $Q(r+1)$ by the following decisions.
\begin{enumerate}
\item[a.] For round $r=0$, all the agents {simultaneously} decide their first edge (with tail $s_i$), and get queued there, defining $Q(1)$.
\item[b.] Now, we focus on the agents who need to make a decision:\\
For round $r\geq 1$, let set $M_{H(r)}$ be every agent $i$ on top of a queue $Q_e(r)$ s.t. $F(e)\neq\emptyset$ and the head of edge $e$ is not her sink $s_i^\ast$. All the agents in $M_{H(r)}$ simultaneously pop and decide their next edges $e'\in F(e)$, queuing in the corresponding waiting lists $Q_{e'}(r+1)$, and thus defining $Q(r+1)$.
\item[c.] For round $r\geq 1$, every agent $i$ on top of a queue $Q_e(r)$ s.t. the head of edge $e$ is $i$'s sink $s_i^\ast$ pops and exits the system with {total-delay} $C_i=r$. If $F(e)=\emptyset$ and the head is not her sink, she exists with total-delay $C_i=\infty$.
\end{enumerate}
For round $r\geq 0$,
the set of all possible histories is denoted by $\mathcal{H}(r)$.\qedef
\end{definition}

\begin{definition}[Game-tree]
The game-tree $\Gamma=(\mathcal{H}, \mathcal{Q})$ is defined as follows.
\begin{itemize}
\item Nodes are all possible histories: $\mathcal{H}=\bigcup_{r\geq 0} \mathcal{H}(r)$. History $H(0)$ is the root. 
\item For any round $r\geq 0$, a transition from history $H(r)$ to $(H(r),Q(r+1))$ exists in $\mathcal{Q}$ if and only if $Q(r+1)$ is a consecutive configuration for $H(r)$.
\end{itemize}
Given game-tree $\Gamma=(\mathcal{H}, \mathcal{Q})$ and history $H(r)\in\mathcal{H}$, we define $\Gamma[H(r)]$ the subtree rooted on $H(r)$. It also defines the subgame starting from $H(r)$.\qedef
\end{definition}

The game-tree is a (possibly infinite) tree graph which entirely describes how the game can be played. On each node/history, agents make their decisions, branching into a consecutive history, until all agents reached their sink.

\begin{definition}[Strategy]
For agent $i$, a strategy $\sigma_i$ is a function which maps every history $H(r)$ in which agent $i$  has to decide a new edge (i.e. agent $i$ belongs to $M_{H(r)}$) to her next edge.
Let $\Sigma_i$ denote the set of possible strategies for $\sigma_i$.
A strategy-profile $(\sigma_1,\ldots,\sigma_n)\in\Sigma_1\times\ldots\times\Sigma_n$, which we denote in bold by $\bm{\sigma}\in\bm{\Sigma}$, defines a strategy $\sigma_i$ for every agent $i$, along with an adversary profile $\bm{\sigma}_{-i}$.

A \emph{strategical} (resp. \emph{trivial}) agent is an agent for who there is more than one (resp. exactly one) $s_i-s_i^\ast-$path. 
\qedef
\end{definition}
A strategy could be a huge or infinite object, representing it entirely is intractable. Instead, one may  assume that an agent has her own \emph{polynomial-time} algorithm for deciding her next edge, function of current configuration or history. 

A \frog\ ends in a finite number of rounds if an agent never decides to go where there is no path to her sink and never visits the same vertex/edge twice (e.g. if $G$ is acyclic or as a strategy).

\begin{definition}[Induced path-profile]
Given strategy-profile $\bm{\sigma}$, agents play their strategies on the game.
If the game ends in a finite number of rounds, then it generates one finite sequence of edges per agent: a unique path-profile $\bm{\pi}(\bm{\sigma})$.\qedef
\end{definition}
\begin{definition}[Total-delay]
For every agent $i$ and path-profile $\bm{\pi}$, 
let $C_i(\bm{\pi})$ be the total-delay of agent $i$: 
the round when she reaches her sink $s_i^\ast$ and exits the system.
It can be defined as $C_i(\bm{\pi})=\sum_{\ell=1}^{|\pi_i|}w_i(\bm{\pi},\ell)$ 
where $w_{i}(\bm{\pi},\ell)\in\mathbb{N}_{\geq 1}$ is agent $i$'s position in waiting list $Q_{\pi_i(\ell)}$ when she enters edge $\pi_i(\ell)$.

Function $C_i$ extends to strategy-profiles $\bm{\sigma}$ by mapping from path-profile $\bm{\pi}(\bm{\sigma})$ if the game ends in a finite number of rounds. Otherwise, an agent who reaches her goal in a finite number of rounds has this number as a total delay,
and an agent who doesn't get delay $C_i(\bm{\sigma})=\infty$.
Furthermore, function $C_i(\bm{\sigma} \mid H(r))$ is agent $i$'s total-delay when strategies $\bm{\sigma}$ are played starting from history $H(r)$, that is in subgame $\Gamma[H(r)]$.\qedef
\end{definition}

It is worth noting that a path-profile $(\pi_1,\ldots,\pi_n)$ can be mapped in polynomial time to total delays $(C_1(\bm{\pi}),\ldots,C_n(\bm{\pi}))$ by mean of a Dijkstra-style pop-the-next-event algorithm similar to \cite[Prop 2.2]{harks2016competitive} or \cite[Th. 2]{ismaili2017}.

\begin{definition}[Subgame perfect equilibrium (SPE)]
Given a \frog, an SPE is a strategy-profile 
$(\sigma_1,\ldots,\sigma_n)$
such that in game-tree $\Gamma=(\mathcal{H}, \mathcal{Q})$, 
for any history $H(r)\in\mathcal{H}$, any agent $i$ in $M_{H(r)}$ and any deviation $\sigma_i'\in\Sigma_i$, one has:
$$C_i(\bm{\sigma}\mid H(r)) \quad \leq\quad  C_i(\sigma_i',\bm{\sigma}_{-i}\mid H(r)).\vspace*{-8mm}$$
\qedef
\end{definition}
In other words, in any subgame, strategy-profile $\bm{\sigma}$ is a pure Nash equilibrium. 
When an SPE is played, the game always ends in a finite number of rounds, since it is in players' interests. 

\begin{example}
We depict below the game-tree of Example 1.
Every rectangle represents the set of deciding agents $M_{H(r)}$.
On every transition, we represent the edges decided.
Total delays are 2-vectors on leaves. SPEs are in bold.
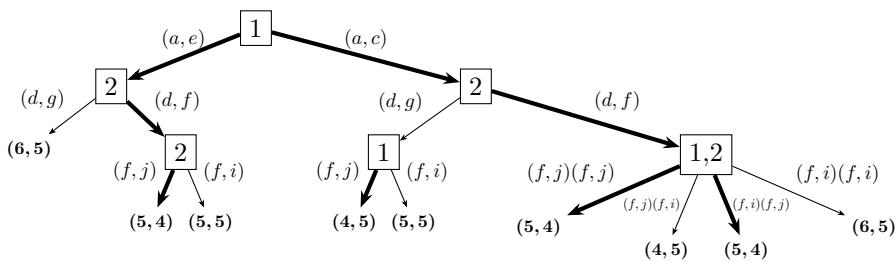
\begin{center}
\begin{tikzpicture}
\draw[]
	node[rectangle,draw=black] (A) {1}
	node[rectangle,draw=black, below right= 3mm and 25mm of A] (B) {2}
	node[rectangle,draw=black, below left= 3mm and 15mm of A] (C) {2}
	node[rectangle,draw=black, below right= 4mm and 25mm of B] (D) {1,2}
	node[rectangle,draw=black, below left= 4mm and 8mm of B] (E) {1}
	node[rectangle,draw=black, below right= 4mm and 5mm of C] (F) {2}
	node[below left= 4mm and 5mm of C,scale=0.67] (G) {$\bm{(6,5)}$}
	node[below right= 5mm and 15mm of D,scale=0.67] (H) {$\bm{(6,5)}$}
	node[below right= 8mm and -2mm of D,scale=0.67] (I) {$\bm{(5,4)}$}
	node[below left= 8mm and -2mm of D,scale=0.67] (J) {$\bm{(4,5)}$}
	node[below left= 5mm and 15mm of D,scale=0.67] (K) {$\bm{(5,4)}$}
	node[below right= 5mm and -2mm of E,scale=0.67] (L) {$\bm{(5,5)}$}
	node[below left= 5mm and -2mm of E,scale=0.67] (M) {$\bm{(4,5)}$}
	node[below right= 5mm and -2mm of F,scale=0.67] (N) {$\bm{(5,5)}$}
	node[below left= 5mm and -2mm of F,scale=0.67] (O) {$\bm{(5,4)}$};
\draw[]
	(A) edge[-{Stealth[scale=0.6]},ultra thick] node[above,scale=0.75]{$(a,c)$} (B) 
	(A) edge[-{Stealth[scale=0.6]},ultra thick] node[above,scale=0.75]{$(a,e)$} (C) 
	(B) edge[-{Stealth[scale=0.6]}, ultra thick] node[above right,scale=0.75]{$(d,f)$} (D) 
	(B) edge[-{Stealth[scale=0.6]}] node[above left,scale=0.75]{$(d,g)$}(E) 
	(C) edge[-{Stealth[scale=0.6]}, ultra thick] node[above right,scale=0.75]{$(d,f)$}(F) 
	(C) edge[-{Stealth[scale=0.6]}] node[above left,scale=0.75]{$(d,g)$}(G)
	(D) edge[-{Stealth[scale=0.6]}] node[above right,scale=0.75]{$(f,i)(f,i)$} (H)
	(D) edge[-{Stealth[scale=0.6]}, ultra thick] node[right,scale=0.5]{$(f,i)(f,j)$}(I) 
	(D) edge[-{Stealth[scale=0.6]}] node[left,scale=0.5]{$(f,j)(f,i)$}(J) 
	(D) edge[-{Stealth[scale=0.6]}, ultra thick] node[above left,scale=0.75]{$(f,j)(f,j)$} (K)
	(E) edge[-{Stealth[scale=0.6]}] node[above right,scale=0.75]{$(f,i)$}(L) 
	(E) edge[-{Stealth[scale=0.6]}, ultra thick] node[above left,scale=0.75]{$(f,j)$}(M)  
	(F) edge[-{Stealth[scale=0.6]}] node[above right,scale=0.75]{$(f,i)$} (N) 
	(F) edge[-{Stealth[scale=0.6]},ultra thick] node[above left,scale=0.75]{$(f,j)$} (O);
\end{tikzpicture}
\end{center}
\end{example}

\subsection{Discussion on Tie-breaking, Actions, Information and Existence}

In the simultaneous-action setting defined above,
several agents may simultaneously queue on a same waiting list $Q_e$, 
during a \emph{same round}, hence a tie breaking rule for who queues first is required. Two exist in literature:
\begin{enumerate}[leftmargin=11mm]
\item[(RO)] There is an overall order $\succ$ on agents, and those with the higher priority queue first. However, in this case, an SPE is not guaranteed to exist.\footnote{Here, we do not consider mixed strategies.} (See e.g. \cite[Fig. 2]{cao2017arxiv}.)
\item[(RE)] For every vertex $u$, a priority $\succ_u$ is defined on incoming edges. If all agents share same sink, an SPE is guaranteed to exist \cite[Sec. 3]{cao2017arxiv}. 
However, if sinks are different, an SPE is not guaranteed.\footnote{See e.g. the pursuer-follower counter-example in \cite[Fig. 1]{ismaili2017}. The same construction, but with priorities on incoming edges, is strategically equivalent. Moreover, since the only two strategical agents decide everything simultaneously on the first round, pure Nash equilibrium and SPE coincide (as an empty concept there).} 
\end{enumerate}

The simultaneousness of decisions in each round, induces imperfect information:
agents in set $M_{H(r)}$ do not know what the others decide at round $r$. 
In order to obtain a setting with sequential actions and perfect information, which guarantees existence of an SPE by backward induction, we use the same setting as above, with the following change that slices rounds into agent turns.
\begin{enumerate}[leftmargin=11mm]
\item[(RR)] Agents in $M_{H(r)}$ (instead of deciding simultaneously) take turn w.r.t. an agent order $1\succ\ldots\succ n$ for deciding their next edge and queue there. Hence, only one agent per (slice of) round makes a decision and acts. 
\end{enumerate}

\subsection{Computational problems studied}

The main issue lies in settling the complexity\footnote{We assume as common knowledge: decision and function problem, length function, complexity classes P, NP, PSPACE and FPSPACE, many-to-one reduction, hardness, completeness and decision problems 3SAT and QSAT (also known as TQBF).} of 
computing a subgame perfect equilibrium (when it always exists),
or of deciding whether one exists (if it may not).
While the length function of \frog s is polynomial in $|V|$ and $n$,
a strategy $\sigma$ is by definition an output of intractable length.
Therefore, when SPE is guaranteed, computation of its induced path profile is a more formalizable problem. We explore this sequence of problems, where $\mathcal{R}$ refers to rules (RO) and (RR):
\begin{itemize}[leftmargin=3.5cm]
\setlength{\itemsep}{0.4em}
\item[\textsc{FRoG/$\mathcal{R}$/Br}:~~~~~]
Given a \frog , an agent $i$, an adversary path-profile $\bm{\pi}_{-i}$ and a threshold $\theta\in\mathbb{N}$, decide whether there is a path $\pi_i$ with total delay $C_i(\pi_i,\bm{\pi}_{-i})\leq\theta$.
\item[\textsc{FRoG/$\mathcal{R}$/Win}:~~~]
Given a \frog , an agent $i$ and a threshold $\theta\in\mathbb{N}$, decide whether there is a strategy $\sigma_i$ that guarantees her total delay $C_i(\sigma_i,\bm{\sigma}_{-i})\leq\theta$ for any adversary strategy $\bm{\sigma}_{-i}$.
\item[\textsc{FRoG/RO/Exist}:~]
Given a \frog , does it admit an SPE?
\item[\textsc{FRoG/RR/Find}:]
Given a \frog , find a path-profile induced by an SPE.
\end{itemize}

\section{The Complexity of Winning Strategies}

In this section, we show that deciding whether an agent has a winning strategy (one that grants her total delay below some threshold) is PSPACE-complete under simultaneous actions (RO) and sequential actions (RR).

\begin{theorem}\label{th:win}
\textsc{FRoG/$\mathcal{R}$/Win} is PSPACE-complete under (RO) and (RR).
\end{theorem}

The remainder of this section consists in the proof of Th. \ref{th:win}. In Def. \ref{def:loosener}, we first define gadget games  that allow preference inversions which break Bellman's principle and cause computational intractability.
Then, for didactic reasons, in Lem. \ref{lem:br}, we consider a best-response decision problem related to \cite[Th. 3]{ismaili2017}, which lets us introduce smoothly an original many-one reduction from 3SAT that is then generalized to QSAT by reusing the same structure and gadget-games.

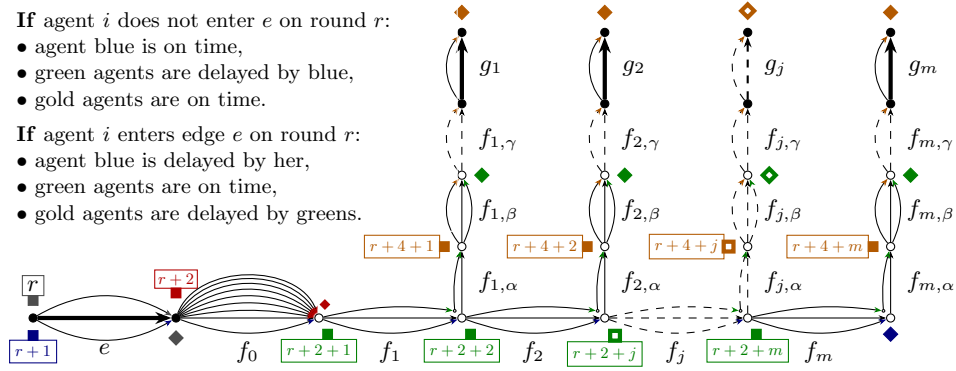
\begin{figure}[t]
\centering
\begin{tikzpicture}[scale=0.95]
\node[circle,draw=black,inner sep=0.0em,fill=black] (e) at (0,0) {$~$};
\node[circle,draw=black,inner sep=0.0em,fill=black] (f0) at (2,0) {$~$};
\node[circle,draw=black,inner sep=0.0em] (f1) at (4,0) {$~$};
\node[circle,draw=black,inner sep=0.0em] (f2) at (6,0) {$~$};
	\node[circle,draw=black,inner sep=0.0em] (f2prime) at (5.9,0.1) {};
\node[circle,draw=black,inner sep=0.0em] (f3) at (8,0) {$~$};
	\node[circle,draw=black,inner sep=0.0em] (f3prime) at (7.9,0.1) {};
\node[circle,draw=black,inner sep=0.0em] (fm) at (10,0) {$~$};
	\node[circle,draw=black,inner sep=0.0em] (fmprime) at (9.9,0.1) {};
\node[circle,draw=black,inner sep=0.0em] (fn) at (12,0) {$~$};
	\node[circle,draw=black,inner sep=0.0em] (fnprime) at (11.9,0.1) {};
\node[circle,draw=black,inner sep=0.0em] (f2a) at (6,1) {$~$};
\node[circle,draw=black,inner sep=0.0em] (f3a) at (8,1) {$~$};
\node[circle,draw=black,inner sep=0.0em] (fma) at (10,1) {$~$};
\node[circle,draw=black,inner sep=0.0em] (fna) at (12,1) {$~$};
\node[circle,draw=black,inner sep=0.0em] (f2b) at (6,2) {$~$};
\node[circle,draw=black,inner sep=0.0em] (f3b) at (8,2) {$~$};
\node[circle,draw=black,inner sep=0.0em] (fmb) at (10,2) {$~$};
\node[circle,draw=black,inner sep=0.0em] (fnb) at (12,2) {$~$};
\node[circle,draw=black,inner sep=0.0em,fill=black] (f2c) at (6,3) {$~$};
\node[circle,draw=black,inner sep=0.0em,fill=black] (f3c) at (8,3) {$~$};
\node[circle,draw=black,inner sep=0.0em,fill=black] (fmc) at (10,3) {$~$};
\node[circle,draw=black,inner sep=0.0em,fill=black] (fnc) at (12,3) {$~$};
\node[circle,draw=black,inner sep=0.0em,fill=black] (f2d) at (6,4) {$~$};
\node[circle,draw=black,inner sep=0.0em,fill=black] (f3d) at (8,4) {$~$};
\node[circle,draw=black,inner sep=0.0em,fill=black] (fmd) at (10,4) {$~$};
\node[circle,draw=black,inner sep=0.0em,fill=black] (fnd) at (12,4) {$~$};
\draw[] 
	(e) edge[-{Stealth[scale=0.6]},ultra thick] node[below=2mm]{$e$} (f0)
	(f0) edge[-{Stealth[scale=0.6]}] node[below=2mm]{$f_0$} (f1)
	(f1) edge[-{Stealth[scale=0.6]}] node[below=2mm]{$f_1$} (f2)
	(f2) edge[-{Stealth[scale=0.6]}] node[below=2mm]{$f_2$} (f3)
	(f3) edge[-{Stealth[scale=0.6]},dashed] node[below=2mm]{$f_{j}$} (fm)
	(fm) edge[-{Stealth[scale=0.6]}] node[below=2mm]{$f_m$} (fn);
\draw[] 
	(f2) edge[-{Stealth[scale=0.6]}] node[right=1mm]{$f_{1,\alpha}$}(f2a)
	(f2a) edge[-{Stealth[scale=0.6]}] node[right=1mm]{$f_{1,\beta}$}(f2b)
	(f2b) edge[-{Stealth[scale=0.6]},dashed] node[right=1mm]{$f_{1,\gamma}$} (f2c)
	(f2c) edge[-{Stealth[scale=0.6]},ultra thick] node[right=1mm]{$g_1$} (f2d);
\draw[] 
	(f3) edge[-{Stealth[scale=0.6]}] node[right=1mm]{$f_{2,\alpha}$}(f3a)
	(f3a) edge[-{Stealth[scale=0.6]}] node[right=1mm]{$f_{2,\beta}$}(f3b)
	(f3b) edge[-{Stealth[scale=0.6]},dashed] node[right=1mm]{$f_{2,\gamma}$} (f3c)
	(f3c) edge[-{Stealth[scale=0.6]},ultra thick] node[right=1mm]{$g_2$} (f3d);
\draw[] 
	(fm) edge[-{Stealth[scale=0.6]},dashed] node[right=1mm]{$f_{j,\alpha}$} (fma)
	(fma) edge[-{Stealth[scale=0.6]},dashed] node[right=1mm]{$f_{j,\beta}$} (fmb)
	(fmb) edge[-{Stealth[scale=0.6]},dashed] node[right=1mm]{$f_{j,\gamma}$} (fmc)
	(fmc) edge[-{Stealth[scale=0.6]},dashed, thick] node[right=1mm]{$g_j$} (fmd);
\draw[] 
	(fn) edge[-{Stealth[scale=0.6]}] node[right=1mm]{$f_{m,\alpha}$}(fna)
	(fna) edge[-{Stealth[scale=0.6]}] node[right=1mm]{$f_{m,\beta}$}(fnb)
	(fnb) edge[-{Stealth[scale=0.6]},dashed] node[right=1mm]{$f_{m,\gamma}$} (fnc)
	(fnc) edge[-{Stealth[scale=0.6]},ultra thick] node[right=1mm]{$g_m$} (fnd);
\node[above = 1mm of e, rectangle,draw=black!70,fill=black!70,inner sep=0.2em] {~};
\node[rectangle, draw=black!70,above = 2.3mm of e, inner sep=0.2em] {$r$};
\node[ below = 1mm of f0, diamond,draw=black!70,fill=black!70,inner sep=0.15em] {~};
\draw[draw=black!70] 
	(e) edge[-{Stealth[scale=0.6]}, bend left = 30] (f0);
\draw[draw=black!50!blue, fill=black!50!blue] 
	node[below = 1mm of e, rectangle,draw=black!50!blue,fill=black!50!blue,inner sep=0.2em] {~}
	node[below = 2.3mm of e,scale=0.67, draw=black!50!blue, text=black!50!blue] {$r+1$}
	node[below = 0.5mm of fn, diamond,draw=black!50!blue,fill=black!50!blue,inner sep=0.15em] {~}
	(e) edge[-{Stealth[scale=0.6]}, bend right = 30] (f0)
	(f0) edge[-{Stealth[scale=0.6]}, bend right = 20] (f1)
	(f1) edge[-{Stealth[scale=0.6]}, bend right = 20] (f2)
	(f2) edge[-{Stealth[scale=0.6]}, bend right = 20] (f3)
	(f3) edge[-{Stealth[scale=0.6]}, bend right = 20,dashed] (fm)
	(fm) edge[-{Stealth[scale=0.6]}, bend right = 20] (fn);
\draw[draw=black!50!green] 
	node[rectangle,below right = 1mm and 0mm of f1,draw=black!50!green,fill=black!50!green,inner sep=0.2em] {~}
	node[diamond,right = 1mm of f2b,draw=black!50!green,fill=black!50!green,inner sep=0.15em] {~}
	node[rectangle,below = 2.2mm of f1,draw=black!50!green, text=black!50!green,scale=0.67] {$r+2+1$}
	(f1) edge[-{Stealth[scale=0.6]}, bend left = 10] (f2prime)
	(f2prime) edge[-{Stealth[scale=0.6]}, bend left =10] (f2a)
	(f2a) edge[-{Stealth[scale=0.6]}, bend right = 25] (f2b);
\draw[draw=black!50!green] 
	node[rectangle,below right = 1mm and 0mm of f2,draw=black!50!green,fill=black!50!green,inner sep=0.2em] {~}
	node[diamond,right = 1mm of f3b,draw=black!50!green,fill=black!50!green,inner sep=0.15em] {~}
	node[rectangle,below = 2.2mm of f2,draw=black!50!green,text=black!50!green,scale=0.67] {$r+2+2$}
	(f2) edge[-{Stealth[scale=0.6]}, bend left = 10] (f3prime)
	(f3prime) edge[-{Stealth[scale=0.6]}, bend left =10] (f3a)
	(f3a) edge[-{Stealth[scale=0.6]}, bend right = 25] (f3b);
\draw[draw=black!50!green] 
	node[rectangle,below right = 1mm and 0mm of f3,draw=black!50!green,ultra thick,inner sep=0.2em] {~}
	node[diamond,right = 1mm of fmb,draw=black!50!green,ultra thick,inner sep=0.15em] {~}
	node[rectangle,below = 2.5mm of f3,draw=black!50!green,text=black!50!green,scale=0.67] {$r+2+j$}
	(f3) edge[-{Stealth[scale=0.6]}, bend left = 10,dashed] (fmprime)
	(fmprime) edge[-{Stealth[scale=0.6]}, bend left =10,dashed] (fma)
	(fma) edge[-{Stealth[scale=0.6]}, bend right = 25,dashed] (fmb);
\draw[draw=black!50!green] 
	node[rectangle,below right = 1mm and 0mm of fm,draw=black!50!green,fill=black!50!green,inner sep=0.2em] {~}
	node[diamond,right = 1mm of fnb,draw=black!50!green,fill=black!50!green,inner sep=0.15em] {~}
	node[rectangle,below = 2.2mm of fm,draw=black!50!green,text=black!50!green,scale=0.67] {$r+2+m$}
	(fm) edge[-{Stealth[scale=0.6]}, bend left = 10] (fnprime)
	(fnprime) edge[-{Stealth[scale=0.6]}, bend left =10] (fna)
	(fna) edge[-{Stealth[scale=0.6]}, bend right = 25] (fnb);
\draw[draw=black!30!orange] 
	node[rectangle,left = 1mm of f2a,draw=black!30!orange,fill=black!30!orange,inner sep=0.2em] {~}
	node[diamond,above = 1mm of f2d,draw=black!30!orange,fill=black!30!orange,inner sep=0.15em] {~}
	node[rectangle,left = 2.3mm of f2a,draw=black!30!orange,text=black!30!orange,scale=0.67] {$r+4+1$}
	(f2a) edge[-{Stealth[scale=0.6]}, bend left = 40] (f2b)
	(f2b) edge[-{Stealth[scale=0.6]}, bend left = 40, dashed] (f2c)
	(f2c) edge[-{Stealth[scale=0.6]}, bend left = 40] (f2d);
\draw[draw=black!30!orange] 
	node[rectangle,left = 1mm of f3a,draw=black!30!orange,fill=black!30!orange,inner sep=0.2em] {~}
	node[diamond,above = 1mm of f3d,draw=black!30!orange,fill=black!30!orange,inner sep=0.15em] {~}
	node[rectangle,left = 2.3mm of f3a,draw=black!30!orange,text=black!30!orange,scale=0.67] {$r+4+2$}
	(f3a) edge[-{Stealth[scale=0.6]}, bend left = 40] (f3b)
	(f3b) edge[-{Stealth[scale=0.6]}, bend left = 40, dashed] (f3c)
	(f3c) edge[-{Stealth[scale=0.6]}, bend left = 40] (f3d);
\draw[draw=black!30!orange] 
	node[rectangle,left = 1mm of fma,draw=black!30!orange,ultra thick,inner sep=0.2em] {~}
	node[diamond,above = 1mm of fmd,draw=black!30!orange,ultra thick,inner sep=0.15em] {~}
	node[rectangle,left = 2.5mm of fma,draw=black!30!orange,text=black!30!orange,scale=0.67] {$r+4+j$}
	(fma) edge[-{Stealth[scale=0.6]}, bend left = 40, dashed] (fmb)
	(fmb) edge[-{Stealth[scale=0.6]}, bend left = 40, dashed] (fmc)
	(fmc) edge[-{Stealth[scale=0.6]}, bend left = 40,dashed] (fmd);
\draw[draw=black!30!orange] 
	node[rectangle,left = 1mm of fna,draw=black!30!orange,fill=black!30!orange,inner sep=0.2em] {~}
	node[diamond,above = 1mm of fnd,draw=black!30!orange,fill=black!30!orange,inner sep=0.15em] {~}
	node[rectangle,left = 2.3mm of fna,draw=black!30!orange,text=black!30!orange,scale=0.67] {$r+4+m$}
	(fna) edge[-{Stealth[scale=0.6]}, bend left = 40] (fnb)
	(fnb) edge[-{Stealth[scale=0.6]}, bend left = 40, dashed] (fnc)
	(fnc) edge[-{Stealth[scale=0.6]}, bend left = 40] (fnd);
\draw[draw=black!30!red] 
	node[above = 2mm of f0, rectangle,draw=black!30!red,fill=black!30!red,inner sep=0.2em] {~}
	node[above right = 1mm and 0mm of f1, diamond,draw=black!30!red,fill=black!30!red,inner sep=0.1em] {~}
	node[above = 3.3mm of f0, text=black!30!red,draw=black!30!red,inner sep=0.2em,scale=0.67] {$r+2$}
	(f0) edge[-{Stealth[scale=0.6]}, bend left = 10] (f1)
	(f0) edge[-{Stealth[scale=0.6]}, bend left = 20] (f1)
	(f0) edge[-{Stealth[scale=0.6]}, bend left = 30] (f1)
	(f0) edge[-{Stealth[scale=0.6]}, bend left = 40] (f1)
	(f0) edge[-{Stealth[scale=0.6]}, bend left = 50] (f1)
	(f0) edge[-{Stealth[scale=0.6]}, bend left = 60] (f1)
	(f0) edge[-{Stealth[scale=0.6]}, bend left = 70] (f1);
\node[rectangle, draw=white, text width=6.4cm,scale=0.9] at (2.8,3.6)  {\textbf{If} agent $i$ does not enter $e$ on round $r$:\\ $\bullet$ agent blue is on time,\\ $\bullet$ green agents are delayed by blue,\\ $\bullet$ gold agents are on time.};
\node[rectangle, draw=white, text width=6.4cm,scale=0.9] at (2.8,2.0) {\textbf{If} agent $i$ enters edge $e$ on round $r$:\\ $\bullet$ agent blue is delayed by her,\\ $\bullet$ green agents are on time,\\ $\bullet$ gold agents are delayed by greens.};
\end{tikzpicture}
\caption{Loosener for (RO) from trigger $(e,r)$ to consequences $(g_1,r_{g_1}),\ldots,(g_m,r_{g_m})$: Starting from edge $e$, there is a path of $m+1$ edges $e, f_0,f_1,\ldots,f_m$. Besides, from any edge $f_j$ ($1\leq j\leq m$), there is a path $f_j,f_{j,\alpha},f_{j,\beta},f_{j,\gamma_1},f_{j,\gamma_2},\ldots,g_j$ to consequence $g_j$. (How many $f_{j,\gamma_\ell}$ edges is specified later.) Agents are as follows. Every agent only has one possible path from his source to his sink. We represent their sources, starting times and sinks by respectively squares, a number near the square and a diamond. A \emph{dark} agent only travels edge $e$ on round $r$. A \emph{blue} agent travels from edge $e$ on round $r+1$ to edge $f_m$, in the best case crossing edge $f_j$ on round $r+2+j$. An arbitrary number of \emph{red} agents all try to only cross edge $f_0$ on same round $r+2$. On every path $f_j,f_{j,\alpha},f_{j,\beta}$, there is a different \emph{green} agent traveling it, starting from round $r+2+j$ on edge $f_j$. Finally, on every path $f_{j,\beta},f_{j,\gamma_1},f_{j,\gamma_2},\ldots,g_j$, there is a \emph{gold} agent traveling it, starting on edge $f_{j,\beta}$ from round $r+4+j$, and in the best case arriving on edge $g_j$ on round $r_{g_j}$. (Hence, there are $r_{g_j}-(r+5+j)$ edges $f_{j,\gamma_\ell}$.) The overall tiebreaking order on agents is defined so that: $\text{blue}\succ\text{green}\succ\text{gold}\succ\text{red}\succ\text{dark}\succ\text{agent } i$.}
\label{fig:loosener}
\end{figure}
\begin{definition}[Loosener]\label{def:loosener}
Given a \frog, an agent $i$, 
a \emph{trigger} couple $(e,r)$ of an edge and a round,
and $m$ \emph{consequence} couples $(g_1,r_{g_1}),\ldots,(g_m,r_{g_m})$ of an edge and a round,
a \emph{loosener} is a piece of \frog\ defined under (RO) as in Figure \ref{fig:loosener},
with the minor assumption that $r_{g_j}\geq r+5+j$, for any $1\leq j\leq m$.\qedef
\end{definition}
\begin{lemma}\label{lem:loosener}
In a loosener, if agent $i$ reaches trigger-edge $e$ on trigger-round $r$,
then her delay on any consequence-edge $g_j$ in consequence-round $r_{g_j}$ is decreased by one round. (An agent with higher priority than $i$ and that was arriving on $g_j$ at $r_{g_j}$ is delayed by one.) Furthermore, the loosener is not a shortcut for $i$: her cost for going through it from round $r$ or later is arbitrarily large.
\end{lemma}
\begin{proof}[Lemma \ref{lem:loosener}, under (RO)]
If agent $i$ enters edge $e$ on round $r$, she gets queued before agent blue (but after dark) since she arrives one round earlier than blue. Hence, blue is delayed one round, and then much more by red agents. Therefore, there are no collisions between blue and the $m$ green agents, who are then on time. Consequently, when a green starts from edge $f_j$ on round $r+2+j$, then she enters edge $f_{j,\beta}$ on round $r+4+j$ and induces a delay on the corresponding gold agent who enters edge $g_j$ on round $r_{g_j}+1$. 

If agent $i$ does not enter edge $e$ on round $r$, then blue enters edge $f_0$ on round $r+2$ before the red agents (by priority). Because of every collision blue and the greens on edges $f_j$ at times $r+2+j$, and because of priority $\text{blue}\succ\text{green}$, every green agent is delayed by one round. Consequently, no green agent collides with any gold agent. Therefore, every gold agent enters his edge $g_j$ on round $r_{g_j}$. 

If agent $i$ tries to go through the loosener from round $r$ or later (in order to take a shortcut to some consequence edge), because of dark's priority on her, $i$ enters edge $f_0$ at round $r+2$ or later. Unfortunately, the arbitrarily large number of red agents have priority over her and she gets queued after them.
\qed
\end{proof}

\begin{lemma}\label{lem:br}
\textsc{FRoG/$\mathcal{R}$/Br} is NP-complete under (RO) and (RR).
\end{lemma}

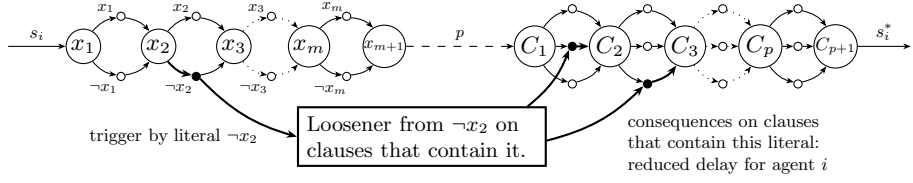
\begin{figure}[t]
\centering
\begin{tikzpicture}
\node[circle,draw=black,inner sep=0.1em] (x1) at (0,0) {${x_1}$};
	\node[circle,draw=black,inner sep=0.0em] (x1a) at (0.5,-0.4) {$~$};
	\node[circle,draw=black,inner sep=0.0em] (x1b) at (0.5,+0.4) {$~$};
\node[circle,draw=black,inner sep=0.1em] (x2) at (1,0) {${x_2}$};
	\node[circle,draw=black,inner sep=0.0em,fill=black] (x2a) at (1.5,-0.4) {$~$};
	\node[circle,draw=black,inner sep=0.0em] (x2b) at (1.5,+0.4) {$~$};
\node[circle,draw=black,inner sep=0.1em] (x3) at (2,0) {${x_3}$};
	\node[circle,draw=black,inner sep=0.0em] (x3a) at (2.5,-0.4) {$~$};
	\node[circle,draw=black,inner sep=0.0em] (x3b)at (2.5,+0.4) {$~$};
\node[circle,draw=black,inner sep=0.1em] (xm) at (3,0) {${x_m}$};
	\node[circle,draw=black,inner sep=0.0em] (xma) at (3.5,-0.4) {$~$};
	\node[circle,draw=black,inner sep=0.0em] (xmb) at (3.5,+0.4) {$~$};
\node[circle,draw=black,inner sep=0.0em,scale=0.72] (xm1) at (4,0) {${x_{m+1}}$};
\draw[] 
	(x1) edge[-{Stealth[scale=0.6]},bend right=20] node[below=1mm,scale=0.67]{$\neg x_1$} (x1a)
	(x1a) edge[-{Stealth[scale=0.6]},bend right=20] (x2)
	(x2) edge[-{Stealth[scale=0.6]},bend right=20,thick] node[below=1mm,scale=0.67]{$\neg x_2$} (x2a)
	(x2a) edge[-{Stealth[scale=0.6]},bend right=20] (x3)
	(x3) edge[-{Stealth[scale=0.6]},bend right=20,dotted] node[below=1mm,scale=0.67]{$\neg x_3$} (x3a)
	(x3a) edge[-{Stealth[scale=0.6]},bend right=20,dotted] (xm)
	(xm) edge[-{Stealth[scale=0.6]},bend right=20] node[below=1mm,scale=0.67]{$\neg x_m$}(xma)
	(xma) edge[-{Stealth[scale=0.6]},bend right=20] (xm1);
\draw[] 
	(x1) edge[-{Stealth[scale=0.6]},bend left=20] node[above=0.5mm,scale=0.67]{$~x_1$} (x1b)
	(x1b) edge[-{Stealth[scale=0.6]},bend left=20] (x2)
	(x2) edge[-{Stealth[scale=0.6]},bend left=20] node[above=0.5mm,scale=0.67]{$~x_2$} (x2b)
	(x2b) edge[-{Stealth[scale=0.6]},bend left=20] (x3)
	(x3) edge[-{Stealth[scale=0.6]},bend left=20,dotted] node[above=0.5mm,scale=0.67]{$~x_3$}(x3b)
	(x3b) edge[-{Stealth[scale=0.6]},bend left=20,dotted] (xm)
	(xm) edge[-{Stealth[scale=0.6]},bend left=20] node[above=0.5mm,scale=0.67]{$~x_m$} (xmb)
	(xmb) edge[-{Stealth[scale=0.6]},bend left=20] (xm1);
\node[circle,draw=black,inner sep=0.1em] (c1) at (6,0) {${C_1}$};
	\node[circle,draw=black,inner sep=0.0em] (c1a) at (6.5,-0.5) {$~$};
	\node[circle,draw=black,inner sep=0.0em] (c1b) at (6.5,+0.5) {$~$};
	\node[circle,draw=black,inner sep=0.0em,fill=black] (c1c) at (6.5, 0.0) {$~$};
\node[circle,draw=black,inner sep=0.1em] (c2) at (7,0) {${C_2}$};
	\node[circle,draw=black,inner sep=0.0em,fill=black] (c2a) at (7.5,-0.5) {$~$};
	\node[circle,draw=black,inner sep=0.0em] (c2b) at (7.5,+0.5) {$~$};
	\node[circle,draw=black,inner sep=0.0em] (c2c) at (7.5, 0.0) {$~$};
\node[circle,draw=black,inner sep=0.1em] (c3) at (8,0) {${C_3}$};
	\node[circle,draw=black,inner sep=0.0em] (c3a) at (8.5,-0.5) {$~$};
	\node[circle,draw=black,inner sep=0.0em] (c3b) at (8.5,+0.5) {$~$};
	\node[circle,draw=black,inner sep=0.0em] (c3c) at (8.5, 0.0) {$~$};
\node[circle,draw=black,inner sep=0.1em] (cp) at (9,0) {${C_p}$};
	\node[circle,draw=black,inner sep=0.0em] (cpa) at (9.5,-0.5) {$~$};
	\node[circle,draw=black,inner sep=0.0em] (cpb) at (9.5,+0.5) {$~$};
	\node[circle,draw=black,inner sep=0.0em] (cpc) at (9.5, 0.0) {$~$};
\node[circle,draw=black,inner sep=0.0em,scale=0.75] (cp1) at (10,0) {${C_{p+1}}$};
\draw[] 
	(xm1) edge[-{Stealth[scale=0.6]}, dashed] node[above,scale=0.67]{$p$} (c1);
\draw[] 
	(c1) edge[-{Stealth[scale=0.6]},bend right=20] (c1a)
	(c1a) edge[-{Stealth[scale=0.6]},bend right=20] (c2)
	(c2) edge[-{Stealth[scale=0.6]},bend right=20] (c2a)
	(c2a) edge[-{Stealth[scale=0.6]},bend right=20,thick] (c3)
	(c3) edge[-{Stealth[scale=0.6]},bend right=20,dotted] (c3a)
	(c3a) edge[-{Stealth[scale=0.6]},bend right=20,dotted] (cp)
	(cp) edge[-{Stealth[scale=0.6]},bend right=20] (cpa)
	(cpa) edge[-{Stealth[scale=0.6]},bend right=20] (cp1);
\draw[] 
	(c1) edge[-{Stealth[scale=0.6]},bend left=20] (c1b)
	(c1b) edge[-{Stealth[scale=0.6]},bend left=20] (c2)
	(c2) edge[-{Stealth[scale=0.6]},bend left=20] (c2b)
	(c2b) edge[-{Stealth[scale=0.6]},bend left=20] (c3)
	(c3) edge[-{Stealth[scale=0.6]},bend left=20,dotted] (c3b)
	(c3b) edge[-{Stealth[scale=0.6]},bend left=20,dotted] (cp)
	(cp) edge[-{Stealth[scale=0.6]},bend left=20] (cpb)
	(cpb) edge[-{Stealth[scale=0.6]},bend left=20] (cp1);
\draw[] 
	(c1) edge[-{Stealth[scale=0.6]}] (c1c)
	(c1c) edge[-{Stealth[scale=0.6]},thick] (c2)
	(c2) edge[-{Stealth[scale=0.6]}] (c2c)
	(c2c) edge[-{Stealth[scale=0.6]}] (c3)
	(c3) edge[-{Stealth[scale=0.6]},dotted] (c3c)
	(c3c) edge[-{Stealth[scale=0.6]},dotted] (cp)
	(cp) edge[-{Stealth[scale=0.6]}] (cpc)
	(cpc) edge[-{Stealth[scale=0.6]}] (cp1);
\draw[] 
	(-1,0) edge[-{Stealth[scale=0.6]}] node[above,scale=0.75]{$s_i$} (x1)
	(cp1) edge[-{Stealth[scale=0.6]}] node[above,scale=0.75]{$s_i^\ast$} (11,0);
\node[rectangle,draw=black,inner sep=0.4em,scale=0.9,thick,text width=3.4cm] (loosen) at (4.5,-1.2) {Loosener from $\neg x_2$ on clauses that contain it.};
\node[scale=0.75] at (1.2,-1.2) {trigger by literal $\neg x_2$};
\node[scale=0.75,text width=3.9cm] at (8.7,-1.3) {consequences on clauses that contain this literal:  reduced delay for agent $i$};
\draw[] 
	(x2a) edge[-{Stealth[scale=0.6]},bend right=15,thick] (loosen)
	(loosen) edge[-{Stealth[scale=0.6]},bend right=15,thick] (c1c)
	(loosen) edge[-{Stealth[scale=0.6]},bend right=15,thick] (c2a);
\end{tikzpicture}
\caption{Reduction from decision problem 3SAT to \textsc{FRoG/$\mathcal{R}$/Br} under (RO) and (RR): Let a 3SAT instance be defined by a list of $m$ binary variables $x_1,\ldots,x_m$ and a list of $p$ 3-clauses $C_1,\ldots,C_p$. A 3-clause is a disjunction of three literals (e.g. $C_j=x_1\vee\neg x_2\vee\neg x_3$). We reduce it to a \textsc{FRoG/$\mathcal{R}$/Br} instance (digraph, agents, priority and threshold). In the digraph, there is a sequence of $m+1$ nodes $x_1,\ldots,x_m,x_{m+1}$, followed by a sequence of $p+1$ nodes $C_1,\ldots,C_p,C_{p+1}$. Between any nodes $x_k$ and $x_{k+1}$, there are two two-edge paths. (Each path represents a choice of valuation for variable $x_k$: true or false.) From node $x_{m+1}$ to node $C_1$, there is one path with $p$ edges. Between any nodes $C_j$ and $C_{j+1}$, there are three two-edge paths that represent the three literals of clause $C_j$. Agent $i$ aims at traveling from source node $x_1$ to sink node $C_{p+1}$; firstly through $m$ choices of valuation for variables, towards node $x_{m+1}$ and then $C_1$; secondly through one literal per node $C_j$, until she reaches node $C_p$. Traveling the first edge of any choice of valuation/literal for variable $x_k$ during round $r=3(k-2)$ loosens on every clause $C_j$ that contains this literal, the arrival time of an agent with higher priority than $i$, from round $3m+p+2j$ to round $3m+p+2j+1$. 
Traveling from vertices $x_1$ to $C_1$ takes constant time $3m+p$. The cost for traveling from vertices $C_j$ to $C_{j+1}$ is two are three, depending on whether a literal was loosened.
The question is whether agent $i$ can travel from vertex $x_1$ to $C_{p+1}$ in less than $\theta=3m+3p$ rounds.}
\label{fig:3SAT}
\end{figure}

\begin{proof}[Lemma \ref{lem:br}, under (RO)]
A path $\pi_i$ with total delay below $\theta$ is a yes-certificate that can be verified in polynomial-time, hence \textsc{FRoG/$\mathcal{R}$/Br} belongs to NP. We show hardness by a reduction from decision problem 3SAT, entirely described in Fig. \ref{fig:3SAT}, and proved below, first under rules (RO).

(yes$\Rightarrow$yes)
Assume a satisfactory instantiation of variables for the 3SAT instance, and let agent $i$ travel from vertices $x_1$ to $x_{m+1}$ accordingly. (It takes two rounds on every trigger and one round on the next edge, hence $3m$ rounds in total.) Since every clause $C_j$ is satisfied by at least one literal (out of three), it means that between $C_j$ and $C_{j+1}$, at least one path has its cost loosened from three to two, path that agent $i$ takes on right time $3m+p+2j$. Therefore her total delay is $3m+3p$.    

(yes$\Leftarrow$yes)
Assume a path with total delay no larger than $3m+3p$, hence equal to it. It means that for every clause $C_j$, at least one delay was loosened (on the right time) by a successful choice of path/instantiation between vertices $x_1$ and $x_{m+1}$, which gives us an instantiation that satisfies the 3SAT instance.
\qed
\end{proof}

\begin{proof}[Generalization of Lemma \ref{lem:loosener} and \ref{lem:br} from rule (RO) to rule (RR)]
Both proofs generalize to rule (RR) by using the same construction. For (RR), it suffices to keep the same overall order over agents, and to observe that since agent $i$ is the only one to be strategical\footnote{The others have singleton strategy-sets: a single path from source to sink.}, no simultaneous decisions occur under (RO), hence rules (RO) and (RR) are here strategically equivalent.
\qed
\end{proof}

We are now ready to state the proof of Th. \ref{th:win} under rules (RO) and (RR). 
\begin{figure}[t]
\centering
\begin{tikzpicture}

\node[circle,draw=black,inner sep=0.1em] (x1) at (0,0) {${x_1}$};
	\node[circle,draw=black,inner sep=0.0em] (x1a) at (0.5,-0.4) {$~$};
	\node[circle,draw=black,inner sep=0.0em] (x1b) at (0.5,+0.4) {$~$};
\node[circle,draw=black,inner sep=0.1em] (x2) at (1,0) {${x_2}$};
\node[circle,draw=black,inner sep=0.1em] (x3) at (2,0) {${x_3}$};
	\node[circle,draw=black,inner sep=0.0em] (x3a) at (2.5,-0.4) {$~$};
	\node[circle,draw=black,inner sep=0.0em] (x3b)at (2.5,+0.4) {$~$};
\node[circle,draw=black,inner sep=0.1em] (xm) at (3,0) {${x_m}$};
\node[circle,draw=black,inner sep=0.0em,scale=0.72] (xm1) at (4,0) {${x_{m+1}}$};
\draw[] 
	(x1) edge[-{Stealth[scale=0.6]},bend right=20,thick] (x1a)
	(x1a) edge[-{Stealth[scale=0.6]},bend right=20]  (x2)
	(x3) edge[-{Stealth[scale=0.6]},bend right=20,thick] (x3a)
	(x3a) edge[-{Stealth[scale=0.6]},bend right=20] (xm);
\draw[] 
	(x1) edge[-{Stealth[scale=0.6]},bend left=20,thick] (x1b)
	(x1b) edge[-{Stealth[scale=0.6]},bend left=20] (x2)
	(x3) edge[-{Stealth[scale=0.6]},bend left=20,thick] (x3b)
	(x3b) edge[-{Stealth[scale=0.6]},bend left=20] (xm);
\draw[]
	(x2) edge[-{Stealth[scale=0.6]}] (x3)
	node[circle,draw=black,fill=white,inner sep=-0.1em] (x2a) at (1.35,0) {$~$}
	node[circle,draw=black,fill=white,inner sep=-0.1em] (x2b) at (1.55,0) {$~$};
\draw[]
	(xm) edge[-{Stealth[scale=0.6]}] (xm1)	
	node[circle,draw=black,fill=white,inner sep=-0.1em] (xma) at (3.35,0) {$~$}
	node[circle,draw=black,fill=white,inner sep=-0.1em] (xmb) at (3.52,0) {$~$};
	
\node[circle,draw=black,inner sep=0.1em] (y1) at (0,2) {${y_1}$};
\node[circle,draw=black,inner sep=0.1em] (y2) at (1,2) {${y_2}$};
	\node[circle,draw=black,inner sep=0.0em] (y2a) at (1.5,+1.6) {$~$};
	\node[circle,draw=black,inner sep=0.0em] (y2b) at (1.5,+2.4) {$~$};
\node[circle,draw=black,inner sep=0.1em] (y3) at (2,2) {${y_3}$};
\node[circle,draw=black,inner sep=0.1em] (ym) at (3,2) {${y_m}$};
	\node[circle,draw=black,inner sep=0.0em] (yma) at (3.5,+1.6) {$~$};
	\node[circle,draw=black,inner sep=0.0em] (ymb) at (3.5,+2.4) {$~$};
\node[circle,draw=black,inner sep=0.0em,scale=0.72] (ym1) at (4,2) {${y_{m+1}}$};
\draw[] 
	(y2) edge[-{Stealth[scale=0.6]},bend right=20,thick] (y2a)
	(y2a) edge[-{Stealth[scale=0.6]},bend right=20] (y3)
	(ym) edge[-{Stealth[scale=0.6]},bend right=20,thick] (yma)
	(yma) edge[-{Stealth[scale=0.6]},bend right=20] (ym1);
\draw[] 
	(y2) edge[-{Stealth[scale=0.6]},bend left=20,thick] (y2b)
	(y2b) edge[-{Stealth[scale=0.6]},bend left=20] (y3)
	(ym) edge[-{Stealth[scale=0.6]},bend left=20,thick] (ymb)
	(ymb) edge[-{Stealth[scale=0.6]},bend left=20] (ym1);
\draw[]
	(y1) edge[-{Stealth[scale=0.6]}] (y2)
	node[circle,draw=black,fill=white,inner sep=-0.1em] (y1a) at (0.35,+2.0) {$~$}
	node[circle,draw=black,fill=white,inner sep=-0.1em] (y1b) at (0.55,+2.0) {$~$};
\draw[] 
	(y3)  edge[-{Stealth[scale=0.6]}] (ym)
	node[circle,draw=black,fill=white,inner sep=-0.1em] (y3a) at (2.35,+2.0) {$~$}
	node[circle,draw=black,fill=white,inner sep=-0.1em] (y3b)at (2.55,+2.0) {$~$};
\begin{scope}[shift={(0,0)}]
\node[circle,draw=black,inner sep=0.1em] (c1) at (5,0) {${C_1}$};
	\node[circle,draw=black,inner sep=0.0em] (c1a) at (5.5,-0.5) {~};
	\node[circle,draw=black,inner sep=0.0em] (c1b) at (5.5,+0.5) {~};
	\node[circle,draw=black,inner sep=0.0em] (c1c) at (5.5, 0.0) {~};
\node[circle,draw=black,inner sep=0.1em] (c2) at (6,0) {${C_2}$};
	\node[circle,draw=black,inner sep=0.0em] (c2a) at (6.5,-0.5) {~};
	\node[circle,draw=black,inner sep=0.0em] (c2b) at (6.5,+0.5) {~};
	\node[circle,draw=black,inner sep=0.0em] (c2c) at (6.5, 0.0) {~};
\node[circle,draw=black,inner sep=0.1em] (c3) at (7,0) {${C_3}$};
	\node[circle,draw=black,inner sep=0.0em] (c3a) at (7.5,-0.5) {~};
	\node[circle,draw=black,inner sep=0.0em] (c3b) at (7.5,+0.5) {~};
	\node[circle,draw=black,inner sep=0.0em] (c3c) at (7.5, 0.0) {~};
\node[circle,draw=black,inner sep=0.1em] (cp) at (8,0) {${C_p}$};
	\node[circle,draw=black,inner sep=0.0em] (cpa) at (8.5,-0.5) {~};
	\node[circle,draw=black,inner sep=0.0em] (cpb) at (8.5,+0.5) {~};
	\node[circle,draw=black,inner sep=0.0em] (cpc) at (8.5, 0.0) {~};
\node[circle,draw=black,inner sep=0.0em,scale=0.75] (cp1) at (9,0) {${C_{p+1}}$};
\draw[] 
	(xm1) edge[-{Stealth[scale=0.6]}, dotted,thick] node[above,scale=0.67]{$p$} (c1);
\draw[] 
	(c1) edge[-{Stealth[scale=0.6]},bend right=20] (c1a)
	(c1a) edge[-{Stealth[scale=0.6]},bend right=20] (c2)
	(c2) edge[-{Stealth[scale=0.6]},bend right=20] (c2a)
	(c2a) edge[-{Stealth[scale=0.6]},bend right=20] (c3)
	(c3) edge[-{Stealth[scale=0.6]},bend right=20] (c3a)
	(c3a) edge[-{Stealth[scale=0.6]},bend right=20] (cp)
	(cp) edge[-{Stealth[scale=0.6]},bend right=20] (cpa)
	(cpa) edge[-{Stealth[scale=0.6]},bend right=20] (cp1);
\draw[] 
	(c1) edge[-{Stealth[scale=0.6]},bend left=20] (c1b)
	(c1b) edge[-{Stealth[scale=0.6]},bend left=20] (c2)
	(c2) edge[-{Stealth[scale=0.6]},bend left=20] (c2b)
	(c2b) edge[-{Stealth[scale=0.6]},bend left=20] (c3)
	(c3) edge[-{Stealth[scale=0.6]},bend left=20] (c3b)
	(c3b) edge[-{Stealth[scale=0.6]},bend left=20] (cp)
	(cp) edge[-{Stealth[scale=0.6]},bend left=20] (cpb)
	(cpb) edge[-{Stealth[scale=0.6]},bend left=20] (cp1);
\draw[] 
	(c1) edge[-{Stealth[scale=0.6]}] (c1c)
	(c1c) edge[-{Stealth[scale=0.6]}] (c2)
	(c2) edge[-{Stealth[scale=0.6]}] (c2c)
	(c2c) edge[-{Stealth[scale=0.6]}] (c3)
	(c3) edge[-{Stealth[scale=0.6]}] (c3c)
	(c3c) edge[-{Stealth[scale=0.6]}] (cp)
	(cp) edge[-{Stealth[scale=0.6]}] (cpc)
	(cpc) edge[-{Stealth[scale=0.6]}] (cp1);
\end{scope}
\node[rectangle,draw=black,text width=1.5cm,scale=0.9,align=center,dotted] at (-1.5,0) {$\exists$-player};
\node[rectangle,draw=black,text width=1.5cm,scale=0.9,align=center,dotted] at (-1.5,2) {$\forall$-player};
\draw[] 
	(-0.7,0) edge[-{Stealth[scale=0.6]}] node[above,scale=0.75]{$s_\exists$} (x1)
	(cp1) edge[-{Stealth[scale=0.6]}] node[above,scale=0.75]{$s_\exists^\ast$} (10,0.0);
\draw[] 
	(-0.7,2) edge[-{Stealth[scale=0.6]}] node[above,scale=0.75]{$s_\forall$} (y1)
	(ym1) edge[-{Stealth[scale=0.6]}] node[above,scale=0.75]{$s_\forall^\ast$} (4.8,2);
\node[rectangle,dashed,draw=black,below right =-0.4mm and -0.5mm of x1a,scale=0.6](ln1){Loosen $\neg x_1$};
\node[rectangle,dashed,draw=black,below right =-0.4mm and -0.5mm of y2a,scale=0.6](ln2){Loosen $\neg x_2$};
\node[rectangle,dashed,draw=black,below right =-0.4mm and -0.5mm of x3a,scale=0.6](ln3){Loosen $\neg x_3$};
\node[rectangle,dashed,draw=black,below right =-0.4mm and -0.5mm of yma,scale=0.6](lnm){Loosen $\neg x_m$};
\node[rectangle,dashed,draw=black,above right =-0.4mm and -0.5mm of x1b,scale=0.6](lp1){Loosen $x_1$};
\node[rectangle,dashed,draw=black,above right =-0.4mm and -0.5mm of y2b,scale=0.6](lp2){Loosen $x_2$};
\node[rectangle,dashed,draw=black,above right =-0.4mm and -0.5mm of x3b,scale=0.6](lp3){Loosen $x_3$};
\node[rectangle,dashed,draw=black,above right =-0.4mm and -0.5mm of ymb,scale=0.6](lpm){Loosen $x_m$};
\draw[dotted]
	(ln1) edge[-{Stealth[scale=0.6]},bend right=18] (cpa)
	(lp2) edge[-{Stealth[scale=0.6]},bend left=25] (cpb)
	(ln3) edge[-{Stealth[scale=0.6]},bend right=15] (c1a)
	(ln3) edge[-{Stealth[scale=0.6]},bend right=15] (c3a)
	(lpm) edge[-{Stealth[scale=0.6]},bend left=30] (c1b)
	(lpm) edge[-{Stealth[scale=0.6]},bend left=30] (c2b)
	(lp1) edge[-{Stealth[scale=0.6]},bend left=15] (c1c)
	(ln2) edge[-{Stealth[scale=0.6]},bend left=0] (c2c)
	(ln2) edge[-{Stealth[scale=0.6]},bend left=28] (c3c)
	(lp3) edge[-{Stealth[scale=0.6]},bend right=8] (c2a)
	(lnm) edge[-{Stealth[scale=0.6]},bend left=25] (c3b)
	(lnm) edge[-{Stealth[scale=0.6]},bend left=35] (cpc);
\end{tikzpicture}
\caption{Reduction from decision problem QSAT to \textsc{FRoG/$\mathcal{R}$/Win} : Let a QSAT instance be defined by a list of $m$ quantified binary variables $Q_1x_1,\ldots,Q_mx_m$ with $Q_k\in\{\exists,\forall\}$ and a list of $p$ 3-clauses $C_1,\ldots,C_p$ with literals defined on any variable. 
E.g.: $\exists x_1\forall x_2\exists x_3\forall x_4
(x_1 \vee \neg x_3\vee x_4)\wedge
(\neg x_2 \vee x_3\vee x_4)\wedge
(\neg x_2 \vee\neg x_3 \vee\neg x_4)\wedge
(\neg x_1\vee x_2\vee \neg x_4)$.
We reduce it to a \textsc{FRoG/$\mathcal{R}$/Win} instance similar to Figure \ref{fig:3SAT}, with the following changes and additions. A new agent who goes from vertices $y_1$ to $y_{m+1}$ is created, in order to model the universal player of the formula. If a binary  variable $x_k$ has quantifier $\exists$ (resp. $\forall$), then the same two literal-choosing two-edge paths as in Fig. \ref{fig:3SAT} are created from vertex $x_k$ to $x_{k+1}$ (resp. from vertex $y_k$ to $y_{k+1}$), and a unique path of three edges is created from vertex $y_k$ to $y_{k+1}$ (resp. from vertex $x_k$ to $x_{k+1}$). Similarly to Fig. \ref{fig:3SAT}, when an agent takes a literal-choosing path from $x_k$ to $x_{k+1}$ (or from $y_k$ to $y_{k+1}$), in clauses $C_j$ that contain this literal in a path to $C_{j+1}$, it loosens the delay from three rounds to two, during round $3m+p+2j$. Is there a strategy for the existential player to reach sink $C_{p+1}$ in less than $\theta=3m+3p$, whatever the universal player decides?}\label{fig:winning}
\end{figure}
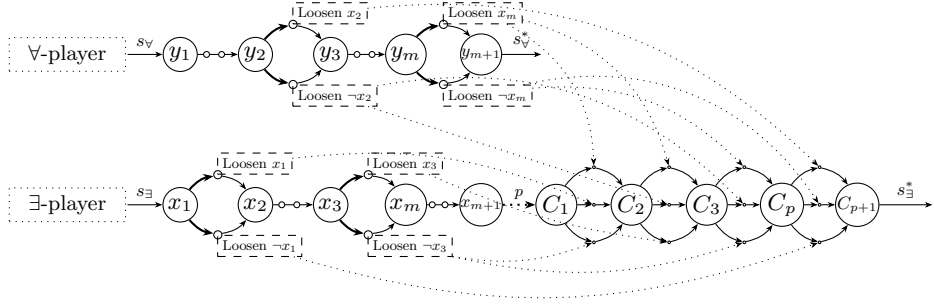
\begin{proof}[Theorem \ref{th:win}]
We show membership to PSPACE by the following minimax algorithm from combinatorial game theory.
Given a \frog\ and a threshold $\theta$,
we must decide whether there is a (winning) strategy $\sigma_i$ for agent $i$ 
that gives her total delay at most $C_i(\sigma_i,\bm{\sigma}_{-i})\leq\theta$ 
for any adversary strategies $\bm{\sigma}_{-i}$.
To do so, we explore game-tree $\Gamma=(\mathcal{H},\mathcal{Q})$ recursively, 
by computing on any subgame $\Gamma[H(r)]$ the best value guaranteed for agent $i$ in $\Gamma[H(r)]$, as follows.
If she is the only agent in $M_{H(r)}$ (to make a decision at this round), then she decides a next edge $e$ that optimizes her total delay according to the values of consecutive subgames $\Gamma((H(r),Q(r+1)))$.
If she does not belong to deciding agents $M_{H(r)}$, then we take the worst value on all consecutive subgames $\Gamma((H(r),Q(r+1)))$.
If there are several agents in $M_{H(r)}$ including her, then a new edge for agent $i$ is evaluated by considering the worst case of what the other agents in $M_{H(r)}$ may decide. 
In this optimization process with alternating minima and maxima, agent $i$ reaches her sink in a polynomial number of edges (because that is her interest), whatever strategies the other agents choose, even the weirdest non-finishing ones (for them). Consequently, we only explore a finite subset of subgames and the (interesting part of the) game tree only has polynomial-depth. Therefore, running  a depth-first search takes (a long time, but) polynomial-space.

We show hardness for PSPACE by a many-one reduction from decision problem QSAT (also known as TQBF), depicted in Figure \ref{fig:winning}.
A QSAT instance is defined by a list of $m$ quantified binary variables $Q_1x_1,\ldots,Q_mx_m$ with $Q_k\in\{\exists,\forall\}$ and a list of $p$ 3-clauses $C_1,\ldots,C_p$ with literals defined on any variable. It asks whether the following formula is true:
$
Q_1x_1,~Q_2x_2,~\ldots,~Q_mx_m,~
C_1\wedge C_2\wedge\ldots\wedge C_p.
$
It is well known that the validity of a QSAT formula can be interpreted as a zero-sum game between an existential player and a universal player \cite[Proof of Th. 4.1]{PAPADIMITRIOU1991}. The formula is true (resp. false) if and only if the existential (resp. universal) player has a winning strategy for instantiating variables.
Here, the universal player is overall indifferent between all the paths from $y_1$ to $y_{m+1}$, and may decide any strategy on universally quantified variables.
Since no simultaneous actions occur in our construct, it holds for both rules (RO) and (RR), by strategical equivalence. 

(yes$\Rightarrow$yes) In the QSAT instance, if the existential player has a winning strategy against the universal player, the existential agent plays it on the path from vertex $x_1$ to $x_{m+1}$, by observing sequentially what the other agent decides from vertex $y_1$ to $y_{m+1}$. Since the formula is satisfied, between any vertex $C_j$ and $C_{j+1}$, there is at least one path loosened on the right round $3m+p+2j$ for the existential agent to travel from vertex $C_1$ to $C_{p+1}$ in $2p$ rounds.

(yes$\Leftarrow$yes) Assume that there is a path strategy for the existential agent from vertex $x_1$ to $x_{m+1}$ such that whatever edges the universal agent decides between $y_1$ and $y_{m+1}$, the delay from vertex $C_1$ to $C_{p+1}$ is $2p$ rounds. Then, between every vertices $C_j,C_{j+1}$, at least one literal-path is loosened. Therefore, the existential player in QSAT has the same winning strategy: the formula is true.\qed
\end{proof}

We now show a deeper negative result on the inapproximability of \frog s under simultaneous and sequential rules (RO) and (RR).

\begin{theorem}\label{th:approx}
It is PSPACE-hard to approximate the optimization version of \textsc{FRoG/$\mathcal{R}$/Win} within $n^{1-\varepsilon}$ for any $\varepsilon>0$, and within any polynomial of $|V|$, under both simultaneous and sequential tiebreaking rules (RO) and (RR). 
\end{theorem}
\begin{proof}[Theorem \ref{th:approx}]
Starting from QSAT, we reuse the same reduction  as for Theorem \ref{th:win} (Figure \ref{fig:winning}), and append one last edge after vertex $C_{p+1}$ that the existential agent should cross to reach her goal. A number $M$ of agents enter this last edge on same round $3m+3p+2$, with higher priority than agent $i$. Hence, either agent $i$ succeeds in having total delay $\text{OPT}=3m+3p+1$, either she fails and obtains total delay $\text{OPT}+M$. The resulting problem contains $|V|=\text{poly}(m,p)$ vertices and $n=\text{poly}(m,p)+M$ agents.

If the QSAT formula is a yes-instance, then the optimum for agent $i$ is $3m+3p+1$, and an $n^{1-\varepsilon}$-approximation algorithm should return a solution within $\text{OPT}(\text{poly}(m,p)+M)^{1-\varepsilon}$, which turns out to be strictly less than $\text{OPT}+M$ for $M$ large enough and still polynomial in $(m,p)$. So the only approximate solution becomes the optimum. If the formula is a no, then total delay is the other one: $3m+3p+1+M$. Consequently, $n^{1-\varepsilon}$-approximation is PSPACE-hard.

Similarly, assume a $q(|V|)$-approximation algorithm for any polynomial $q$. 
If the QSAT formula is a yes-instance, then it should return a solution within $\text{OPT}q(|V|)$, which is strictly less than $\text{OPT}+M$ for a large enough polynomial $M(m,p)$. Therefore, $q(|V|)$-approximation is also PSPACE-hard.\qed
\end{proof}

\section{The Complexity of Subgame Perfect Equilibrium}

In this section, we settle the computational complexity of subgame perfect equilibrium in \frog s
under simultaneous actions (RO) and sequential actions (RR).
\begin{theorem}\label{th:exist}
Decision problem \textsc{FRoG/RO/Exist} is PSPACE-complete,
and function problem \textsc{FRoG/RR/Find} is also PSPACE-hard and in FPSPACE.
\end{theorem}
\begin{proof}[Theorem \ref{th:exist}]
We show membership to PSPACE or FPSPACE by the same algorithm,
which recursively explores game-tree $\Gamma=(\mathcal{H},\mathcal{Q})$ as follows.
In the process, we only approve sequential (resp. simultaneous) decision nodes $H(r)$ where (the) agent(s) in $M_{H(r)}$ can play a best-response (resp. pure Nash equilibrium) that leads to a subgame $\Gamma(H(r),Q(r+1))$ that contains an SPE.
In order to only explore a finite subset of subgames, one can cut the sub-trees where an agent went into a dead-end or where an agent's current delay (that is: current round $r$) is already more than $|E|\times|N|$. 
Indeed, these subtrees are obviously not in the best individual interest of those agents, 
since even a blind strategy can do better whatever the others' strategies.
Consequently, these subtrees cannot contain an SPE. 
Because of this cut, the game tree only has polynomial depth. 
Therefore, a depth-first search takes polynomial-space (but a lot of time).

For function problem \textsc{FRoG/RR/Find}, we show PSPACE-hardness by turning Figure \ref{fig:winning} into a zero-sum game between the existential and universal agents: It suffices to add one last edge $e$ on both strategical agents' paths (reachable in some constant delay from vertices $C_{p+1}$ and $y_{m+1}$). On this edge, the universal agent (who was overall indifferent) is now delayed by one round if and only if the existential agent travels from vertex $C_1$ to $C_{p+1}$ in $2p$ rounds (formula satisfied) and then collides with her on edge $e$.
Since any winning strategy is always played in an SPE, the induced paths indicate whether the formula is true or false.

For decision problem \textsc{FRoG/RO/Exist}, we proceed from Figure \ref{fig:winning} as for \textsc{FRoG/RR/Find}, with the following addition:
If the universal agent is delayed, then she is synchronized to collide with and cancel out counter-example \cite[Fig. 2]{cao2017arxiv} (e.g. by delaying an agent therein). Therefore, if the formula is a yes (resp. no), the universal agent is delayed, the counter-example is (resp. not) canceled and the \textsc{FRoG/RO/Exist} instance is a yes (resp. no).\qed
\end{proof}

\section{Conclusion}

In this paper, for sequential routing games, under simultaneous and sequential tiebreaking rules (RO) and (RR), we settle the computational complexity of the winning strategy problem as PSPACE-complete. PSPACE-hardness of subgame perfect equilibrium is a consequence of the complexity of winning strategies.

Interestingly, under tiebreaking rule (RE), where there are priorities on incoming edges, rather than on agents, the best-response problem turns out P-time tractable (see Th. \ref{th:br:RE} below). This result contrasts with intractability results Lemma \ref{lem:br} and \cite[Th. 3-5]{ismaili2017} which assume a tie-breaking order on agents. Rule (RE) (first proposed in \cite{cao2017arxiv}) turns out to be a sane model that does not favor agents by their IDs. Hence, it will be worth studying further.
\begin{theorem}\label{th:br:RE}
Best-response problem \textsc{FRoG/RE/Br} is in P.
\end{theorem}
\begin{proof}[Theorem \ref{th:br:RE}]
Let us run an initializing event-based Dijkstra-style algorithm on $G=(V,E)$ with adversary paths $\bm{\pi}_{-i}$.
It provides us with congestions $Q_e(r)$ of every queue $e$ and every round $r$, in polynomial-time. 

We now consider augmented graph $\mathcal{G}=(\{s_i\}\cup E,\mathcal{E})$.
Every vertex $v\in V$ is augmented into origin-vertex couple $(u,v)\in E$ (but the source).
In set $\mathcal{E}$, there is an augmented edge from origin-vertex $(u,v)$ (resp. $s_i$) to origin-vertex $(v,w)$ (resp. $(s_i,w)$) if the former's vertex (head) is the later's origin (tail). Costs are \emph{dynamically} defined as follows: 
If agent $i$ decides next edge $e=(v,w)$ from origin-vertex $(u,v)$ during round $r$, 
then she reaches $w$ on round $r+|Q_e(r+1)|-|P_{(u,v,w)}(r+1)|+1$, where $Q_e(r+1)$ is known from the first algorithm (without agent $i$), and agent set $P_{(u,v,w)}(r+1)\subseteq Q_e(r+1)$ are those entering edge $e=(v,w)$ on same round $r+1$ as agent $i$, but from an origin/edge with lower priority than $(u,v)$ (hence $i$ passes them). 

First, let us assume that agent $i$ (whose shortest path we compute) does not incur any queue modification that she causes (as in Fig. \ref{fig:loosener}).
In augmented graph $\mathcal{G}$, despite the fact that costs are dynamic, agent $i$ has no interest in arriving one round later in any waiting list, since it is the maximal rate at which a queue empties. We can then directly run Dijkstra's algorithm on $\mathcal{G}$.

Under tiebreaking rule (RE), any path has same delay for any agent, since delays only depend on when and where they come from, rather than their IDs (like in rules (RO) and (RR)).
Consequently, on a shortest path, agent $i$ can always go at least faster than the chain of queue-modifications that she triggers, by (for instance) taking the same modified path. 
The frog passes Mach one and crosses roads without (new) collisions (Easter egg).
Therefore, on shortest paths, she never incurs any queue modification that she caused, as assumed above.
\qed
\end{proof}

\textbf{Acknowledgments}\quad I am grateful to Ilan Nehama for his enthusiasm, to Silvia De Bellis for proofreading and to the anonymous reviewers for their tedious work.
Following recent open practices, reviews will be appended to a preprint.

\bibliographystyle{alphaCapitalization}
\bibliography{mybib}

\newcommand{\etalchar}[1]{$^{#1}$}
\begin{thebibliography}{CCCW17b}

\bibitem[AAE05]{awerbuch2005price}
Baruch Awerbuch, Yossi Azar, and Amir Epstein.
\newblock The Price of Routing Unsplittable Flow.
\newblock In {\em Proceedings of the thirty-seventh annual ACM symposium on
  Theory of computing}, pages 57--66. ACM, 2005.

\bibitem[AU09]{anshelevich2009equilibria}
Elliot Anshelevich and Satish Ukkusuri.
\newblock Equilibria in Dynamic Selfish Routing.
\newblock In {\em International Symposium on Algorithmic Game Theory}, pages
  171--182. Springer, 2009.

\bibitem[BEDL06]{blum2006routing}
Avrim Blum, Eyal Even-Dar, and Katrina Ligett.
\newblock Routing without regret: On convergence to Nash equilibria of
  regret-minimizing algorithms in routing games.
\newblock In {\em Proceedings of the twenty-fifth annual ACM symposium on
  Principles of distributed computing}, pages 45--52. ACM, 2006.

\bibitem[CCCW17a]{cao2017ec}
Zhigang Cao, Bo~Chen, Xujin Chen, and Changjun Wang.
\newblock {A} {N}etwork {G}ame of {D}ynamic {T}raffic.
\newblock In {\em Proceedings of the 2017 ACM Conference on Economics and
  Computation}, EC '17, pages 695--696. ACM, 2017.

\bibitem[CCCW17b]{cao2017arxiv}
Zhigang Cao, Bo~Chen, Xujin Chen, and Changjun Wang.
\newblock {A} {N}etwork {G}ame of {D}ynamic {T}raffic.
\newblock {\em CoRR}, abs/1705.01784, 2017.

\bibitem[CK05]{christodoulou2005price}
George Christodoulou and Elias Koutsoupias.
\newblock The Price of Anarchy of Finite Congestion Games.
\newblock In {\em Proceedings of the thirty-seventh annual ACM symposium on
  Theory of computing}, pages 67--73. ACM, 2005.

\bibitem[CS11]{chien2011convergence}
Steve Chien and Alistair Sinclair.
\newblock Convergence to approximate Nash equilibria in congestion games.
\newblock {\em Games and Economic Behavior}, 71(2):315--327, 2011.

\bibitem[GJ79]{garey1979computers}
M.R. Garey and D.S. Johnson.
\newblock {\em Computers and Intractability}, volume 174.
\newblock Freeman San Francisco, CA, 1979.

\bibitem[HHP06]{harks2006competitive}
Tobias Harks, Stefan Heinz, and Marc~E Pfetsch.
\newblock Competitive Online Multicommodity Routing.
\newblock In {\em International Workshop on Approximation and Online
  Algorithms}, pages 240--252. Springer, 2006.

\bibitem[HHP09]{harks2009competitive}
Tobias Harks, Stefan Heinz, and Marc~E Pfetsch.
\newblock Competitive Online Multicommodity Routing.
\newblock {\em Theory of Computing Systems}, 45(3):533--554, 2009.

\bibitem[HMRT09]{hoefer2009competitive}
Martin Hoefer, Vahab Mirrokni, Heiko R{\"o}glin, and Shang-Hua Teng.
\newblock Competitive Routing over Time.
\newblock {\em Internet and Network Economics}, pages 18--29, 2009.

\bibitem[HMRT11]{hoefer2011competitive}
Martin Hoefer, Vahab~S Mirrokni, Heiko R{\"o}glin, and Shang-Hua Teng.
\newblock Competitive Routing over Time.
\newblock {\em Theoretical Computer Science}, 412(39):5420--5432, 2011.

\bibitem[HPS{\etalchar{+}}18]{Harks:2018:CPR:3182630.3184137}
Tobias Harks, Britta Peis, Daniel Schmand, Bjoern Tauer, and Laura~Vargas Koch.
\newblock Competitive Packet Routing with Priority Lists.
\newblock {\em ACM Trans. Econ. Comput.}, 6(1):4:1--4:26, March 2018.

\bibitem[HPSVK16]{harks2016competitive}
Tobias Harks, Britta Peis, Daniel Schmand, and Laura Vargas~Koch.
\newblock {C}ompetitive {P}acket {R}outing with {P}riority {L}ists.
\newblock In {\em LIPIcs-Leibniz International Proceedings in Informatics},
  volume~58. Schloss Dagstuhl-Leibniz-Zentrum fuer Informatik, 2016.

\bibitem[Ism17]{ismaili2017}
Anisse Ismaili.
\newblock Routing Games over Time with FIFO Policy.
\newblock In Nikhil R.~Devanur and Pinyan Lu, editors, {\em Web and Internet
  Economics}, pages 266--280, Cham, 2017. Springer International Publishing.

\bibitem[KP99]{koutsoupias1999worst}
Elias Koutsoupias and Christos Papadimitriou.
\newblock Worst-Case Equilibria.
\newblock In {\em STACS}, volume~99, pages 404--413. Springer, 1999.

\bibitem[KS09]{koch2009nash}
Ronald Koch and Martin Skutella.
\newblock Nash Equilibria and the Price of Anarchy for Flows over Time.
\newblock In {\em International Symposium on Algorithmic Game Theory}, pages
  323--334. Springer, 2009.

\bibitem[KS11]{koch2011nash}
Ronald Koch and Martin Skutella.
\newblock Nash Equilibria and the Price of Anarchy for Flows over Time.
\newblock {\em Theory of Computing Systems}, 49(1):71--97, 2011.

\bibitem[MBP17]{monnot2017routing}
Barnab{\'e} Monnot, Francisco Benita, and Georgios Piliouras.
\newblock Routing games in the wild: Efficiency, equilibration and regret.
\newblock In {\em International Conference on Web and Internet Economics},
  pages 340--353. Springer, 2017.

\bibitem[NRTV07]{nisan2007algorithmic}
Noam Nisan, Tim Roughgarden, Eva Tardos, and Vijay~V Vazirani.
\newblock {\em Algorithmic Game Theory}, volume~1.
\newblock Cambridge University Press Cambridge, 2007.

\bibitem[PV10]{papadimitriou2010new}
Christos~H Papadimitriou and Gregory Valiant.
\newblock A New Look at Selfish Routing.
\newblock In {\em ICS}, pages 178--187, 2010.

\bibitem[PY91]{PAPADIMITRIOU1991}
Christos~H. Papadimitriou and Mihalis Yannakakis.
\newblock Shortest paths without a map.
\newblock {\em Theoretical Computer Science}, 84(1):127 -- 150, 1991.

\bibitem[Rou05]{roughgarden2005selfish}
Tim Roughgarden.
\newblock {\em Selfish Routing and the Price of Anarchy}, volume 174.
\newblock Cambridge MIT press, 2005.

\bibitem[Rou09]{roughgarden2009intrinsic}
Tim Roughgarden.
\newblock Intrinsic Robustness of the Price of Anarchy.
\newblock In {\em Proceedings of the forty-first annual ACM symposium on Theory
  of computing}, pages 513--522. ACM, 2009.

\bibitem[RT02]{roughgarden2002bad}
Tim Roughgarden and {\'E}va Tardos.
\newblock How Bad is Selfish Routing?
\newblock {\em Journal of the ACM (JACM)}, 49(2):236--259, 2002.

\bibitem[War52]{wardrop1952road}
John~Glen Wardrop.
\newblock Some Theoretical Aspects of Road Traffic Research.
\newblock {\em Proceedings of the institution of civil engineers},
  1(3):325--362, 1952.

\bibitem[WHK14]{werth2014atomic}
TL~Werth, M~Holzhauser, and SO~Krumke.
\newblock Atomic Routing in a Deterministic Queuing Model.
\newblock {\em Operations Research Perspectives}, 1(1):18--41, 2014.

\end{thebibliography}

\appendix

\section{Reviews from WINE-2018}

\subsection*{Summary of reviews by the author}

I am grateful to the anonymous reviewers for reading the paper, for their tedious work and their constructive comments. I will follow their advises to extend this work, mainly as follows:
\begin{itemize}
\item[+] Additional conceptual contributions are required, for instance: positive results. This may be achieved by studying the tie breaking rule on incoming edges, that was only explored informally in this work. Focusing on simple compact strategies may also reveal interesting.
\item[+] Current results must be written in a more formal, precise and clear way.
\item[+] Related work should be rewritten more clearly, taking into account the name clash with other models of sequential routing games in literature.
\end{itemize}
It goes without saying that all minor comments will be addressed.
For the sake of transparency and valorisation of anonymous work, the full reviews are provided below.

\subsection{Review 1}

In this paper, the authors study a routing game with a time-component where every edge acts as a queue. Every time step, the player in the front of the queue of an edge enters the end of the queue of the next edge he chooses. Sometimes, multiple players want to enter the same queue at the same time step. Various tie breaking rules are then considered: One where tie breaking is determined through an global agent order, one where the tie breaking is over the incoming edges of each vertex, and one where players choose their next edge sequentially over a global agent order (and hence can let their next choice of edge depend on the preceding agents in the order).

Players want to get from their source to their sink as quick as possible, and strategy of a player thus maps a configuration history to a choice of a new adjacent edge whenever the player is on the top of a queue. Hence, an SPE may exist in such games and it definitely exists under the third tie-breaking rule when only one agent moves each turn.

The paper then gives various hardness and completeness results: It shows that under the first and third tie breaking rules it is PSPACE-complete to determine whether a player has a strategy below a given threshold, and it is NP-complete if the strategies of the other players are given. Furthermore, this result is extended (quite straightforwardly) to some inapproximability results with a big inapproximability factors $n^{1-eps}$ and any $poly(|V|)$. The authors also show that existence of an SPE with tie-breaking rule 1 is PSPACE-complete and finding (the resulting outcome of) one in case of rule 3 is PSPACE hard.

Lastly, the authors show in the conclusion section that computing a best response for a player is in P in case the second tie-breaking rule is used. This result is not announced at all in the rest of the paper, and since this is the only "good news" the paper has to offer, I think it is not the best choice to hide this result at the very end of the paper.

All of the hardness results are centered around a reduction gadget which the authors call a "loosener". This "loosener" construction essentially enables the authors to encode various interrelations between various edges in a graph: It encodes that if an agent is able to reach an edge within a particular time step, then he can traverse a set of various other edges one time unit quicker. What I wrote just now is by the way an easy intuition behind the gadget, that I found missing from the paper. I think it would be very helpful to add such a high-level explanation to the paper.

The paper is entirely built around this idea of a "loosener", and using it to reduce a TQBF formula to these routing game problems. The main reduction is not that straightforward though. I am not sure if the result is that surprising or technically deep, but it is reasonably involved.

Regarding the writing quality of the paper, I think that there is much room for improvement. The English language is sometimes used in a way that is a bit weird, with a strange choice of words and terminology sometimes. On the technical side there are further shortcomings in the writing. In the proofs the authors often opt for explaining things in a not-entirely-precise way, and try to avoid formal notation, apparently in the hope that the reader has the same intuition as the authors when reading through their text. I think that is a mistake. An example is the way they define the concept of a "loosener"; the choice of term is not really clearly motivated in the paper, and after introducing the reduction gadget the authors start to use "to loosen" as a verb throughout the text, without really making clear what that means. There is more of this puzzling stuff throughout the paper and I give some more examples in the list below. Also the main reductions are merely "described" o!
r maybe even "sketched", rather than precisely defined, where this description is deferred to pictures and their informal expostition in the captions. These captions are way too long, by the way.

For the discussion of related literature, the authors mention a lot of papers, but at the end of the discussion it was somehow still very unclear to me how the model that is studied here connects more or less to the other works out there. It is not clear to me if this particular model is studied before and how these other papers more or less related to it. Also I want to point out that there is another notion of a "sequential game" in the literature, which the authors do not mention: this concept has been introduced in the paper "The curse of simultaneity" by Paes Leme et al. and subsequently "sequential routing games" have been studied in two previous WINE papers by De Jong and Uetz (2014), and by Correa et al. (2015). The models are different from the one studied here, but they do involve SPEs, and also simply because of the name clash I think this should be discussed.

For all the reasons combined I am not very convinced that the paper is currently in a proper state to be published. I am hoping that the authors could take a more careful look at how they can best write all of their results down in a precise and clear way, and I very much hope for some additional conceptual contributions besides only this central PSPACE-hardness reduction. Some news on the positive side, maybe for the second variant of tie-breaking rule, would be very welcome.

Here is a list of further remarks, suggestions, and corrections.

- behaviors $\rightarrow$ behavior

- travels $\rightarrow$ travel

- "in delays calculations": i think nothing is really being "calculated" here.

- "complexity of games": Is a game necessarily more complex with temporality, or is it just different? Does the non-existence of equilibria necessarily mean that things are more complex? Does hardness of computing equilibria in a class of games imply that the game itself is more complex?

- intractable; the $\rightarrow$ intractable, and the      (also, provide references to these claims)

- a last $\rightarrow$ another

- was related $\rightarrow$ was modeled to be related

- "It is also reminded" sounds strange. Maybe "We also note"?

- "in real life". Also this sounds weird: The paper [MBP17] itself obviously exists "in real life".

- I am not sure why you mention the Canadian Traveler Problem, as any further relation of your results to that problem seems missing.

- "W.l.o.g. on negative results", the "W.l.o.g." here depends on the scope of games that you want to capture, which you did not define.

- The extensive use of footnotes might be a bit too much, and I do think some of these footnotes would be better placed in the main text.

- "Here is the main loop:" Remove this.

- "Strategical" is a bit of an uncommon term. I think "strategic" sounds much more natural.

- Is the notion of a "trivial" player used anywhere, and is it important to any of the results?

- representing it $\rightarrow$ and representing it

- edge, function $\rightarrow$ edge, which is a function

- "Given strategy-profile $\sigma$, agents play their strategies on the game." This is a sort of empty sentence. It does not say anything.

- "by mapping from path profile". I do not understand what this means.

- "has this number". Has which number?

- by mean of $\rightarrow$ by means of

- "Dijkstra-style pop-the-next-event algorithm". This seems very specialized and I do not think you can assume that the reader knows what you mean here. Please elaborate.

- What is $M_{H(r)}$ in Def. 7? (Ignore this if you defined this earlier in the paper.)

- "since it is in players' interests". This is a very vague statement. You use this term "interests" and things like "because it is in her interests" multiple times in the paper and I think you should try to avoid it and be more precise.

- round, induces $\rightarrow$ round induces

- "(as an empty concept here)": what does this mean?

- Define what you mean with "length function" (I don't think this is a very standard term).

- guaranteed $\rightarrow$ guaranteed to exist

- "formalizable" is not the correct word. All that you talk about is entirely formalizable. Perhaps you mean something like "sensible".

- "adversary". Remove this (twice).

- Try to restate Lemma 1 in a precise and clear way, especially the last sentence is confusing to the reader who does not (yet) understand what the purpose of a "loosener" is. Also things like "trigger-edge" and "consequence-edge/round" are a bit mystifying.

- In the definition of a loosener, contrary to how the model was introduced, starting times are suddenly associated to the players, which is confusing.

- The meaning of the bold edges in the picture is not clear. Only much later I realized that a loosener is defined with respect to an existing network, and that these bold edges are edges in that network.

- "two are three"?

- under rules $\rightarrow$ under rule

- "including her". Who is "her"?

- "because that is in her interest". See my remark above.

- Page 10 of the proof of Theorem 1: Remove "even the weirdest non finishing ones (for them)" and "(a long time, but)", since these remarks are unnecessary and distracting. What do you mean with "instantiating variables", and "our construct", and a "winning" strategy against the universal player, and "observing sequentially"? What do you mean with "the same winning strategy: the formula is true."? ("the formula is true" is not a strategy...)

- In the proof of Theorem 2, should not OPT + M be OPT + M + 1?

- I have the feeling that the encoding you use for the decision problems, such as "FRoG/RO/Exist" is not necessary and detrimental to the quality of this paper. I found myself jumping back and forth to the point where these codes are described.

- Proof of Theorem 3: what do you mean with the "same" algorithm? "which" $\rightarrow$ "that". "cut the sub-trees": I think you mean "prune". Be a bit more precise about "when an agent went into a dead-end". "the best individual interest": see above. What is a "blind strategy"? What does it mean for a subtree to "contain" an equilibrium? What do you mean with "this cut" (do you mean something like "this pruning rule"?)? Remove "but a lot of time". What does "delayed" mean? The last paragraph about a particular counter-example in the paper [CCCW17b] is not acceptable as a proof. The definition background behind this counterexample should then be given, and it should be explained why it would be acceptable to just "cancel out" a particular counter-example. Lastly, I would suggest to find a different way of expressing an instance "being a yes" or "a no" as you currently do.

- tiebreaking $\rightarrow$ tie-breaking

- explain what you mean with "sane"

- Proof of Theorem 4: explain what you mean with an "initializing event-based Dijkstra-style algorithm". What do you mean with "augmented". "(but the source)": But the source what? "later's" $\rightarrow$ "latter's". What is this strange sentence about a frog passing Mach one near the end of the proof? Easter egg???

- In the acknowledgments, which of the authors is "I"?

- In the acknowledgments, you announce that you want to somehow publish the reviews. If there is a rebuttal possible to this review, I would be genuinely interested in understanding why the authors would like to do this? I do not really see the point of it.

\subsection{Review 2}

The authors consider sequential routing games, where players decide which arc to use at each node, instead of choosing the entire source-sink-path at once, as in most routing models. The sequential model feels more natural, as players only incur delays when they share arcs with other players at the same time (round).

The results are non-trivial, mostly (but not always) clearly explained, and seem to be correct. However, most of the results are kind of similar to earlier work (Ism17). Also, the are some (minor) shortcomings with respect to organization of the paper as well as langauge.

Detailed comments:

-Related work that may be of interest is 'The curse of sequentiality in routing games, WINE 2015', which also considers the complexity of finding SPE's in routing games.

-Introduction: 'agents incur congestion to one another' You cannot incur something to someone; incurring means having it yourself.

-It took me a while to realize that agent i can exit the loosener either after arc e, or after any of the g-arcs. It would be helpful to explain this more thoroughly, and also show this in Figure 1.

-Theorem 4 is first mentioned (and proven) in the conclusion. I don't think the conclusion is the right place for a proof.

\subsection{Review 3}

The paper studies the computational complexity of a sequential routing game inspired by the dynamic routing model recently introduced in (Car et al., EC 2017). There is an underlying network and the players have different origins and destinations. Each edge operates as a priority queue and has unit length and unit capacity. The new thing is that the players decide on their paths in a hop-by-hop basis, by making local decisions. So, in each round, all agents on top of their queues (as determined by some priority policy) decide on the next hop of their paths based on the history of the game. The individual cost of each player is the number of rounds required to reach his destination.

The paper shows that if (players have different origins and destinations and) priorities are based on player ids, this game is hard to play / solve. Specifically, the paper shows that computing a best-response path, based on perfect knowledge about the paths selected by others, is NP-hard. Moreover, computing any reasonable approximation to a best-response strategy and determining whether the game admits an SPE are PSPACE-hard and PSPACE-complete, respectively. And computing the set of paths induced by an SPE is PSPACE-hard and in FPSPACE. On the other hand, if priorities are based on the edge ids, we can compute a best-response path efficiently (but we do not know about the complexity of the best-response strategy and an SPE).

The model of the sequential routing game considered in the paper is natural and interesting, and the hardness results are definitely good to know. The gadgets used in the hardness proofs are involved (but tailor-made to the specific priority rules, so probably not useful in other similar settings) and the proofs are far from trivial.

My main concern is that the results might not be so interesting for the WINE community because (i) they apply to specific priority rules (player-based priority rules - the general picture might be very different if the priority rules are edge-based), and (ii) they are not accompanied by some positive results (with the exception of Thm. 4, in the conclusions) which would indicate that this is a reasonable model for rational and computationally-restricted players. So I believe that this is a rather borderline paper for WINE. 

A comment for the authors:

- Related work section is not well organized (especially the first and the last paragraphs). If the authors feel obliged to discuss all these models and results, they should organize the section in a way that helps the reader to navigate through them.

\end{document}